\documentclass[amsmath,amssymb,superscriptaddress,nofootinbib,final]{JHEP3}

\RequirePackage{ulem}

\usepackage{amsbsy}
\usepackage{amsmath}
\usepackage{amsthm}

\newcommand{\eps}{\varepsilon}

\newcommand{\be}{\begin{equation}}
\newcommand{\ee}{\end{equation}}
\newcommand{\ba}{\begin{eqnarray}}
\newcommand{\ea}{\end{eqnarray}}

\newcommand{\beq}{\begin{equation}}
\newcommand{\eeq}{\end{equation}}
\newcommand{\beqa}{\begin{eqnarray}}
\newcommand{\eeqa}{\end{eqnarray}}
\newcommand{\nn}{\nonumber}

\newcommand{\hook}{\raisebox{-0.35ex}{\makebox[0.6em][r]
{\scriptsize $-$}}\hspace{-0.15em}\raisebox{0.25ex}{\makebox[0.4em][l]{\tiny
 $|$}}}
\newcommand{\eq}[1]{(\ref{#1})}

\newcommand{\cwedge}[1]{\mathop{\wedge}_{{}^{#1}} }

\newcommand{\even}{{\mathrm{e}}}
\newcommand{\odd}{{\mathrm{o}}}
\newcommand{\KY}{{\scriptscriptstyle\mathrm{KY}}}
\newcommand{\SN}{{\scriptscriptstyle\mathrm{SN}}}

\newcommand{\dm}{n}

\newcommand{\lied}{\mathcal{L}}

\newcommand{\gc}[2]{{[\mspace{-4mu}\lvert#1,#2\rvert\mspace{-4mu}]}}
\newcommand{\bgc}[2]{{\bigl[\mspace{-4.4mu}\lvert#1,#2\rvert\mspace{-4.4mu}\bigr]}}

\newcommand{\efr}{{\underset{\raisebox{0.5ex}{$\scriptscriptstyle\rightarrow$}}{e}\mspace{-2mu}}}
\newcommand{\Xfr}{{\overset{\raisebox{-0.5ex}{$\scriptscriptstyle\rightarrow$}}{X}\mspace{-1mu}}}

\newtheorem{prop}{Proposition}[section]
\newtheorem{lemma}[prop]{Lemma}

\preprint{DAMTP-2011-1}

\title{Commuting symmetry operators of the Dirac equation, Killing--Yano and Schouten-Nijenhuis brackets}
\author{Marco Cariglia\\
Universidade Federal de Ouro Preto, ICEB, Departamento de F\'isica. \\  Campus Morro do Cruzeiro, Morro do Cruzeiro,
35400-000 - Ouro Preto, MG - Brasil \\ E-mail: \email{marco@iceb.ufop.br}}

\author{Pavel Krtou\v{s}\\
Institute of Theoretical physics, Faculty of Mathematics and Physics, Charles University,\\
V Hole\v{s}ovi\v{c}k\'{a}ch 2, 180 00 Prague, Czech Republic\\
E-mail: \email{Pavel.Krtous@utf.mff.cuni.cz}}

\author{David Kubiz\v n\'ak\\
DAMTP, University of Cambridge, Wilberforce Road, Cambridge CB3 0WA, UK\\
E-mail: \email{D.Kubiznak@damtp.cam.ac.uk}
}

\abstract{%
In this paper we derive the most general first-order symmetry operator commuting with the Dirac operator
in all dimensions and signatures. Such an operator splits into Clifford even and Clifford odd parts which are given in terms of odd Killing--Yano and even closed conformal
Killing--Yano inhomogeneous forms respectively. We study commutators of
these symmetry operators and give necessary and sufficient conditions under which they remain of the first-order.
In this specific setting we can introduce a Killing--Yano bracket, a bilinear operation acting on
odd Killing--Yano and even closed conformal
Killing--Yano forms, and demonstrate that it is closely related to the Schouten--Nijenhuis bracket.
An important non-trivial example of vanishing Killing--Yano brackets is given by Dirac symmetry operators generated from the principal conformal Killing--Yano tensor [hep-th/0612029]. We show that among these operators one can find a complete subset of mutually commuting operators. These operators underlie separability of the Dirac equation in Kerr-NUT-(A)dS spacetimes in all dimensions [arXiv:0711.0078].
}


\begin{document}

%%%%%%%%%%%%%%%%%%%%%%%%%%%%%%%%%%%%%%%%%%%%%%%%%%%%%%%%%%%%%%%%%%%%%%%%%%
%%%%%%%%%%%%%%%%%%%%%%%%%%%%%%%%%%%%%%%%%%%%%%%%%%%%%%%%%%%%%%%%%%%%%%%%%%
\clearpage

\section{Introduction}\label{sc:intro}
Symmetries play a central role in modern theoretical physics. They comprise fundamental laws of nature, can be used for classifying solutions,
enable one to simplify otherwise complicated physical problems.
In this paper we shall concentrate on symmetries of fundamental equations in curved spacetime. At the operator level, such symmetries
correspond to the so called {\em symmetry operators}.\footnote{In this paper, by word `operator' we always mean
local (differential) operator.}
These are operators which {\em R-commute} with the operator of the corresponding equation.
Such operators have the property that they map one solution of the equation to another.
For example, let $D$ be a field operator (for concreteness we can think of the Dirac operator) and $\psi$ be a solution of the corresponding field
equation, $D\psi=0$. Then the operator $S$ is a symmetry operator
of $D$, i.e., $R$-commuting with $D$,  if it satisfies
\be\label{jedna0}
DS=RD
\ee
for some operator $R$.
Obviously, $\psi'=S\psi$ satisfy $D\psi'=DS\psi=RD\psi=0$ and hence $\psi'$ is a new solution of the field
equation.
Among all symmetry operators the
{\em commuting operators} play a prominent role. Their eigenvalues yield
 `quantum numbers' characterizing the solution---the `constants of motion'.
Very interesting is the case when one has a {\em complete set of mutually commuting operators} and their common eigenfunctions can be found by separating variables. The corresponding eigenvalues then completely characterize the separated solution and play a role of separation constants.

The problem of separability of various fundamental equations is particularly interesting in curved spacetime.
In fact, a complete theory is known only for the {\em additive separation} of the massive {\em Hamilton--Jacobi} equation \cite{BenentiFrancaviglia:1979, KalninsMiller:1981}.
The objects underlying the separability are Killing vectors and rank-2 Killing tensors.
These symmetries have to satisfy a whole set of algebraic and differential constraints; they have to
form the so called {\em separability structure}. The central role in this theory is played by the {\em symmetric Schouten--Nijenhuis brackets}
\cite{Schouten:1940, Schouten:1954, Nijenhuis:1955}; separability of the Hamilton--Jacobi equation is possible only if one can find a complete set of mutually (Schouten--Nijenhuis)-commuting Killing vectors and rank-2 Killing tensors.\footnote{This requirement is equivalent to the requirement that the corresponding {\em moment maps} (integrals of motion) are functionally independent and mutually Poisson commute.}
Closely related to the additive separation of the Hamilton--Jacobi equation is the theory of {\em multiplicative separation of the Klein--Gordon} equation. The key symmetry is again that of Killing vectors and rank-2 Killing tensors which in addition to form a separability structure have to satisfy additional `quantum' constraints \cite{Carter:1977}.

Much less is known about separability of field equations with spin.
In this paper we discuss the {\em Dirac equation}.
In curved spacetime the study of the subject dates back to the seminal paper of Chandrasekhar who in 1976 separated and decoupled the massive Dirac equation in the Kerr geometry \cite{Chandrasekhar:1976}.
A few years later Carter and McLenaghan \cite{CarterMcLenaghan:1979} were able to demonstrate that behind such separability stands a first-order operator commuting with the Dirac operator---constructed from a Killing--Yano 2-form of Penrose and Floyd \cite{Penrose:1973}.
This discovery stimulated subsequent developments in the study of symmetry operators of the Dirac equation in curved spacetime. In particular, the most general {\em first-order operator} commuting with the Dirac operator
in 4D was constructed by McLenaghan and Spindel \cite{McLenaghanSpindel:1979}. It corresponds to Killing vectors and Killing--Yano symmetries. This work was later extended by Kamran and McLenaghan \cite{KamranMcLenaghan:1984} to {\em R-commuting} symmetry operators, cf. Eq. \eq{jedna0}. Such operators map solutions of the massless Dirac equation to solutions and correspond to  symmetries which are {\em conformal generalizations} of Killing vectors and Killing--Yano tensors. Remarkably, these first-order symmetry operators are sufficient to justify separability of the massless Dirac equation in the whole Plebanski--Demianski class of spacetimes \cite{PlebanskiDemianski:1976}.
This is, however, not true in general. First-order operators are not enough to
completely characterize all Dirac separable systems and one has to consider higher-order symmetry
operators. In fact, Fels and Kamran  \cite{FelsKamran:1990} were able to provide an example
of 4D Lorentzian spin manifold where the Dirac equation separates but the separability is justified by an operator of the second-order.
This motivated people to study higher-order symmetry operators of the Dirac operator, e.g., \cite{McLenaghanEtal:2000}.
It also means that the theory of separability of the Dirac equation must reach outside the realms of the so called {\em factorizable systems} \cite{Miller:1988}, as such systems are fully characterized by first-order symmetry operators.

With recent developments in higher-dimensional gravity, string theory, and various supergravities physicists have become interested in
the Dirac operator in spacetimes of `arbitrary' dimension and signature as well as spacetimes where
the metric is supplemented by additional matter fields (fluxes) which couple to the spinor and modify the Dirac equation.
Consequently, symmetry operators of the (modified) Dirac operator as well as separability of the corresponding Dirac equation have been studied in these more general setups.
In particular,
in their remarkable paper Oota and Yasui  \cite{OotaYasui:2008} were able to
separate the massive Dirac equation in the most general known Kerr-NUT-(A)dS spacetime in all dimensions \cite{ChenEtal:2006cqg}.
Even more generally, separability of the torsion modified Dirac equation was demonstrated in the presence of $U(1)$ and torsion fluxes of the Kerr--Sen geometry and its higher-dimensional generalizations \cite{HouriEtal:2010b} as well as in the most general spherical black hole spacetime of minimal gauged supergravity \cite{Wu:2009a}.
It is the aim of this paper to intrinsically characterize these results.
For simplicity we limit ourselves to the case of a standard Dirac operator but impose no restrictions on
number of spacetime dimensions or the signature. (We believe that an analogous discussion can be performed for flux-modified Dirac operators as well.)
In order to keep formulas relatively simple and calculations tractable in what follows we use the effective formalism of forms developed in \cite{BennTucker:book, HouriEtal:2010a}. For convenience we summarize this formalism in the appendix to which we refer the reader.

It is well known that first-order symmetry operators of the Dirac operator in all dimensions correspond to symmetries associated with {\em conformal Killing--Yano (CKY)} forms.
In $n$ number of dimensions, an (inhomogeneous) CKY form $\omega$ can be written as a sum of its homogeneous $p$-form parts $\omega^{(p)}\in\Omega^p(M)$,
$\omega=\sum_p \omega^{(p)}$, and obeys the following (twistor) equation \cite{Kashiwada:1968}:
\be\label{jedna}
T_a\omega\equiv \nabla_a\omega-\frac{1}{\pi+1}X_a\hook d\omega+\frac{1}{n-\pi+1}e_a\wedge \delta \omega=0\,.
\ee
The operator $T_a$ on the left hand side is called a {\em twistor operator}. Hence, CKY forms are in the kernel of the twistor operator. Such forms are truly special---their covariant derivative splits into the exterior and divergence parts. If in addition the divergence part is zero, $\delta \omega=0$, we have {\em Killing--Yano (KY)} forms, whereas we have {\em closed conformal Killing--Yano (CCKY)} forms when $d\omega=0$.

The following important result on first-order symmetry operators of the Dirac operator is due to
Benn, Charlton, and Kress \cite{BennCharlton:1997, BennKress:2004}. It is valid
in all dimensions  $n$ and arbitrary signature:
\begin{prop}[Symmetry operators]\label{1.1}
The most general first-order symmetry operator $S$ of the Dirac operator $D=e^a\nabla_a$,
i.e. an operator satisfying $DS=RD$ for some operator ${R}$,
is given by
\be
S=S_\omega+\alpha D\,,
\ee
where $\alpha$ is an arbitrary inhomogeneous form, and $S_\omega$, given in terms of an inhomogeneous CKY form $\omega$ obeying \eqref{jedna},  is
\be\label{SOProp1}
S_\omega=X^a\hook\omega\nabla_a+\frac{\pi-1}{2\pi}d\omega-\frac{n-\pi-1}{2(n-\pi)}\delta\omega\;.
\ee
$S_\omega$ obeys
\be
\bgc{D}{S_\omega}\equiv DS_\omega-(\eta S_\omega)D=-\Bigl(\frac{\eta}{n-\pi} \delta\omega\Bigr)D\,.
\ee
\end{prop}
\noindent Apart from an (inevitable) freedom of adding an arbitrary form $\alpha$ this proposition states that in all dimensions and signatures first-order Dirac symmetry operators are in one-to-one correspondence with CKY symmetries.
It generalizes the 4D result of Kamran and McLenaghan \cite{KamranMcLenaghan:1984}.

One of the main goals of the present paper is to shed some light on the result of Oota and Yasui \cite{OotaYasui:2008}.
Namely we want to demonstrate that, similar to four
\cite{CarterMcLenaghan:1979}  and five \cite{Wu:2008b} dimensions,
in any dimension separability of the Dirac equation in the Kerr-NUT-AdS spacetimes is accompanied by the existence of a {\em complete set of mutually commuting first-order operators}.
We show that this set can be chosen to consist of operators corresponding to Killing vectors and CCKY even-forms which can be generated from the {\em principal conformal Killing--Yano (PCKY) tensor} present in the spacetime \cite{KubiznakFrolov:2007,KrtousEtal:2007jhep}.
For this purpose we need to study mutual commutation of operators in the set. We shall prove the result in three steps.
i) Starting from Prop. \ref{1.1} we first establish a uniqueness result for first-order operators commuting with the Dirac operator. We demonstrate that in all dimensions such operators have a unique form. They split into Clifford odd and Clifford even parts and correspond to CCKY even-forms and KY odd-forms, respectively. ii) We study commutators of these operators. (Such commutators trivially commute with the Dirac operator.) We give
sufficient conditions under which these commutators remain of the first-order.
In that case the commutators have to be first-order operators commuting with the Dirac operator and hence of the form studied in the previous step. This allows us
to introduce a {\em Killing--Yano bracket}, a bilinear operation acting on the space of odd KY and even CCKY forms. In particular when the Killing--Yano bracket vanishes the corresponding operators commute. iii) We demonstrate that for the chosen set of symmetries generated from the PCKY tensor in Kerr-NUT-AdS spacetimes all Killing--Yano brackets vanish. Hence, the corresponding operators mutually commute and since they are all `independent' they form a complete set of commuting operators.

As a by-product of our construction we have introduced a {\em Killing--Yano bracket}, a bilinear operation acting on a subspace of odd KY and even CCKY forms. The definition of Killing--Yano brackets can be further extended by considering various (graded) \mbox{(anti-)commuting} linear symmetry operators of the Dirac operator. Such operators are of their own importance; they are relevant for the discussion of the relationship between the Dirac operator level and the corresponding spinning particle description \cite{GibbonsEtal:1993}, they are directly related to the higher-order symmetry symmetry operators of the Dirac equation.
Moreover, studying these operators allows to generalize the Killing--Yano brackets to KY forms of arbitrary degree, and to define graded Killing--Yano brackets for arbitrary KY forms, which for homogeneous forms are directly related to Schouten--Nijenhuis brackets. Using Hodge duality we give a number of conditions under which Killing--Yano brackets of KY and CCKY forms reduce to Schouten--Nijenhuis brackets.
A further possible generalization is to consider $R$-commuting symmetry operators of the Dirac operator and hence to extend the Killing--Yano brackets to arbitrary CKY forms.
 In particular, we give the Killing--Yano bracket of a CKY 1-form and an arbitrary CKY form $\omega$ and show that if the 1-form is KY then this bracket reduces to the Lie derivative of $\omega$.

The fact that Killing--Yano brackets are related to Schouten--Nijenhuis brackets is closely linked to the work of Kastor, Ray, and Traschen \cite{KastorEtal:2007}. In their paper the authors investigated whether KY tensors form a Lie algebra with respect  to the Schouten--Nijenhuis brackets.\footnote{%
It is well known that (symmetric) Killing tensors form an algebra with respect to the symmetric Schouten--Nijenhuis brackets.
}
They demonstrated that this is indeed the case in maximally symmetric spacetimes. On the other hand they were able to find two counter examples of spacetimes where this is not true. Our results provide a natural framework for these investigations. In particular we are able to formulate sufficient conditions under which Killing--Yano brackets generate KY tensors.

The paper is organized as follows. In the next section we establish a uniqueness result for the first-order symmetry operators commuting with the Dirac operator, study commutators of these operators, and introduce Killing--Yano  brackets. In Sec.~\ref{sc:PCKYsym} general results are applied to the tower of hidden symmetries generated from the PCKY tensor in Kerr-NUT-(A)dS spacetimes in all dimensions. In particular, we show that in these spacetimes one can find a complete set of mutually commuting operators for the Dirac equation. In Sec.~\ref{sc:1ordsym} we extend the definition of Killing--Yano brackets to arbitrary KY forms,  define graded Killing--Yano brackets of KY forms by studying other first-order symmetry operators of the Dirac operator, and present uniqueness results for the general form of operators anti-commuting, graded commuting and graded anti-commuting with the Dirac operator. Sec.~\ref{sc:concl} is devoted to discussion and conclusions. In App.~\ref{apx:notation} we gather  our notation and useful identities as well as some technical results
related to Hodge duality, App.~\ref{apx:PCKYsym} summarizes the results on the properties of the tower of hidden symmetries generated from a CCKY 2--form.

%%%%%%%%%%%%%%%%%%%%%%%%%%%%%%%%%%%%%%%%%%%%%%%%%%%%%%%%%%%%%%%%%%%%%%%%%%
%%%%%%%%%%%%%%%%%%%%%%%%%%%%%%%%%%%%%%%%%%%%%%%%%%%%%%%%%%%%%%%%%%%%%%%%%%

\section{Commuting symmetry operators and Killing--Yano brackets}\label{sc:comopKYbr}

%\section{Commuting symmetry operators}\label{commuting}

\subsection{Commuting symmetry operators}\label{ssc:commuting}
In this section we derive the most general first-order operator commuting with the Dirac operator in all dimensions.
We start from the Prop. \ref{1.1} of Benn, Charlton, and Kress \cite{BennCharlton:1997, BennKress:2004}.
It is instructive %(though not necessary but some co-authors insist on it)
to first re-cast this theorem in the following symmetric form:
\begin{prop}
The most general first-order symmetry operator $S$ for the Dirac operator $D$ is given by
\be\label{tL}
  S={\cal S}_\omega+\alpha D\;,
\ee
where $\alpha$ is an arbitrary inhomogeneous form, and ${\cal S}_\omega$ is defined in terms of an inhomogeneous CKY form $\omega$ as
\be\label{Sdef}
  {\cal S}_\omega=\frac12\bigl(e^a\omega+\omega e^a\bigr) \nabla_{\!a}
  +\frac{\pi-1}{2\pi}d\omega-\frac{\dm-\pi-1}{2(\dm-\pi)}\delta\omega\;.
\ee
The operator ${{\cal S}_\omega}$ obeys
\be\label{comSD}
[D, {\cal S}_\omega] = \Bigl(\frac1\pi d\omega_\even-\frac1{\dm-\pi} \delta\omega_\odd \Bigr)\,D\;,
\ee
with ${\omega_\odd}$ and ${\omega_\even}$ being the odd and even parts of ${\omega}$.
\end{prop}
\begin{proof}
This proposition directly follows from Prop. \ref{1.1} if one sets ${\cal S}_\omega=S_\omega+\omega_eD$.
\end{proof}
\noindent
It is now simple to specialize this result to the case of operators that commute with the Dirac operator $D$:
\begin{prop}[Commutation with ${D}$]\label{2.2}
The most general first-order operator ${S}$ which commutes with the Dirac operator $D$, ${[D,S]=0}$, splits into the Clifford even and Clifford odd parts
\be\label{Scom}
  S=S_\even+S_\odd\,,
\ee
where
\begin{align}
    S_\even&= K_{\omega_\odd} \equiv X^a\hook\omega_\odd\nabla_{\!a}
       + \frac{\pi-1}{2\pi}d\omega_\odd\;,
       \quad&&\text{with ${\omega_\odd}$ being an odd KY form}\;,\label{Kdef}\\
    S_\odd&= M_{\omega_\even} \equiv e^a\wedge\omega_\even\nabla_{\!a}
       - \frac{\dm-\pi-1}{2(\dm-\pi)}\delta\omega_\even\;,
       \quad&&\text{with ${\omega_\even}$ being an even CCKY form}\;.\label{Mdef}
\end{align}
\end{prop}
\begin{proof}[Proof] It is straightforward to check that operator ${\cal S}_\omega$, \eqref{Sdef}, can be written as
\begin{equation}\label{S=KM}
    {\cal S}_\omega = K_{\omega_\odd} + M_{\omega_\even}
        - \frac{\dm-\pi-1}{\dm-\pi}\delta\omega_\odd
        + \frac{\pi-1}{2\pi}d\omega_\even\;.
\end{equation}
Clearly, if ${\omega_\odd}$ is a KY form and ${\omega_\even}$ is a CCKY form, we get ${{\cal S}_\omega = K_{\omega_\odd} + M_{\omega_\even}}$ and the right hand side of \eqref{comSD} vanishes. Hence operators \eqref{Kdef} and \eqref{Mdef} commute with the Dirac operator. To prove the uniqueness,
we start from a general symmetry operator $S$, \eqref{tL}, and search for the most general inhomogeneous form $\alpha$ for which this operator commute with the Dirac operator.
We find
\begin{equation}\label{genSDcom}
\begin{split}
    [D,S] &= [D,{\cal S}_\omega+\alpha D]=\Bigl(\frac1\pi d\omega_\even-
      \frac1{\dm{-}\pi}\delta\omega_\odd
      +[D,\alpha]\Bigr)\,D\\
      &=\Bigl[2\bigl(e^a\wedge\alpha_\odd+X^a\hook\alpha_\even\bigr)\nabla_{\!a}
       + \Bigl(\frac1\pi d\omega_\even + d\alpha
       - \frac1{\dm{-}\pi}\delta\omega_\odd-\delta\alpha\Bigr)\Bigr]\,D\;.
\end{split}
\end{equation}
If we require that it should vanish, the square brackets on the right-hand side must be zero in both orders of the derivative. The first-order term implies
\begin{equation}\label{alphatriv}
    (\dm-\pi)\alpha_\odd=0\;,\qquad \pi\alpha_\even =0\;,
\end{equation}
i.e., ${\alpha_\even}$ must be 0-form (which is automatically a KY form) and, in an odd dimension, ${\alpha_\odd}$ must be ${\dm}$-form (which is a CCKY form). Therefore, the ${\alpha}$-term
\begin{equation}
\alpha D=(e^a\wedge\alpha_\even-X^a\hook\alpha_\odd)\nabla_{\!a}
\end{equation}
can be absorbed into ${{\cal S}_\omega}$. The zeroth-order term in \eqref{genSDcom} thus gives
$d\omega_\even=0$, $\delta\omega_\odd=0$\,,
and Eq. \eqref{S=KM} concludes the proof of the proposition.
\end{proof}
\noindent {\em Remark:} The operators $K_{\omega_o}$ and $M_{\omega_e}$ were first constructed in \cite{BennCharlton:1997}. Prop.~\ref{2.2} states that any first-order operator which commutes with the Dirac operator is given in terms of such operators.

%%%%%%%%%%%%%%%%%%%%%%%%%%%%%%%%%%%%%%%%%%%%%%%%%%%%%%%%%%%%%%%%%%%%%%%%%%
%%%%%%%%%%%%%%%%%%%%%%%%%%%%%%%%%%%%%%%%%%%%%%%%%%%%%%%%%%%%%%%%%%%%%%%%%%

\subsection{First-order commutators}\label{ssc:foc}
Let $S_1$ and $S_2$ be arbitrary (not necessarily first-order) operators commuting with the Dirac operator, ${[S_{1,2},D]=0}$. Then, clearly, also ${S=[S_1,S_2]}$ commutes with ${D}$, ${[S,D]=0}$. So in this way we can always construct `new operators' commuting with ${D}$. These new operators, however, will be in general of higher-order in derivatives. For example, when $S_1$ and $S_2$ are of the first-order, their commutator ${S}$ is in general of the second order.\footnote{This is one of the reasons which led people to consider
higher-order symmetry operators of the Dirac operator, see, e.g., \cite{McLenaghanEtal:2000}.}
In what follows we restrict our considerations to the first-order operators commuting with~${D}$, i.e., to operators $K_{\omega_o}$ and $M_{\omega_e}$ given by formulas \eqref{Kdef} and \eqref{Mdef}. More specifically, we concentrate on a special subset of these operators
for which their commutator remains linear in derivatives. The uniqueness result from Prop.~\ref{2.2} then guarantees that in such a case the commutator is again of the form \eqref{Kdef} and \eqref{Mdef}. This allows us to formulate the following:
\begin{prop}[Commutators]\label{pr:algfoop}
Let ${\kappa}$, ${\lambda}$, ${\mu}$ be odd KY forms and ${\alpha}$, ${\beta}$, ${\omega}$ be even CCKY forms
obeying the following algebraic conditions:
\begin{subequations}\label{fordcond}
\begin{gather}{}
    [X^{(a}\hook\kappa,X^{b)}\hook\lambda] = 0\;,\label{fordcondKK}\\
    [X^{(a}\hook\mu,e^{b)}\wedge\omega] = 0\;,\label{fordcondKM}\\
    [e^{(a}\wedge\alpha,e^{b)}\wedge\beta] = 0\;,\label{fordcondMM}
\end{gather}
\end{subequations}
where ${[\ , \ ]}$ means the Clifford commutator on forms. Then and only then
commutators ${[K_\kappa,K_\lambda]}$, ${[K_\mu,M_\omega]}$, and ${[M_\alpha,M_\beta]}$ of the corresponding operators \eqref{Kdef} and \eqref{Mdef} remain of the first-order and one can define `new' odd KY forms ${[\kappa,\lambda]_\KY}$, ${[\alpha,\beta]_\KY}$ and a `new' even CCKY form ${[\mu,\omega]_\KY}$ such that
\begin{equation}\label{KYbrackets0}
\begin{gathered}{}
    [K_\kappa,K_\lambda] = K_{[\kappa,\lambda]_\KY}\;,\quad
    [K_\mu,M_\omega] = M_{[\mu,\omega]_\KY}\;,\quad
    [M_\alpha,M_\beta] = K_{[\alpha,\beta]_\KY}\;.
\end{gathered}
\end{equation}
\end{prop}
\begin{proof}
Let us first consider commutator $[K_\kappa,K_\lambda]$. To simplify our calculation, we denote
\be\label{pomocne}
\kappa^a=X^a\hook \kappa\,,\quad \tilde \kappa=\frac{\pi-1}{2\pi}d\kappa\,,\quad
\lambda^a=X^a\hook \lambda\,,\quad \tilde \lambda=\frac{\pi-1}{2\pi}d\lambda\,.
\ee
Then, using \eqref{Kdef}, we find
\be
K_\kappa K_\lambda =\kappa^a\lambda^b \nabla_a\nabla_b+\bigl[\kappa^a(\nabla_a\lambda^b)+\kappa^b\tilde \lambda+\tilde \kappa\lambda^b\bigr]\nabla_b
+\kappa^a(\nabla_a\tilde \lambda)+\tilde \kappa\tilde \lambda\,,
\ee
and we have
\ba\label{pom}
[K_\kappa, K_\lambda]&=&(\kappa^a\lambda^b-\lambda^a\kappa^b)\nabla_{a}\nabla_{b}+
\Bigl\{\kappa^a(\nabla_a\lambda^b)-\lambda^a(\nabla_a\kappa^b)+[\kappa^b, \tilde \lambda]-[\lambda^b,\tilde \kappa]\Bigr\}\nabla_b
\nonumber\\
&&+\,\kappa^a(\nabla_a\tilde \lambda)-\lambda^a(\nabla_a\tilde \kappa)+[\tilde \kappa,\tilde \lambda]\nonumber\\
&=&
[\kappa^{(a},\lambda^{b)}]\nabla_{a}\nabla_{b}+
\Bigl\{\kappa^a(\nabla_a\lambda^b)-\lambda^a(\nabla_a\kappa^b)+[\kappa^b, \tilde \lambda]-[\lambda^b,\tilde \kappa]\Bigr\}\nabla_b\nonumber\\
&&+\,\mbox{zeroth order terms}\,.
\ea
In the second equality we have used the fact that the antisymmetric part of $\nabla_a\nabla_b$ gives the curvature and hence contributes to the zeroth-order term of the commutator. Obviously, the requirement that the commutator remains of the first-order imposes a condition
$[\kappa^{(a},\lambda^{b)}]=0$, which is equivalent to Eq. \eqref{fordcondKK}.
If this condition is satisfied, $[K_\kappa, K_\lambda]$ must be a linear operator commuting with the Dirac operator and hence of the form described in Prop.~\ref{2.2}. Moreover, using formulas \eqref{usefulProduct1}, we realize that the expression in the curly bracket of \eqref{pom} is even.
Hence, the commutator must be even, and since it is of the form \eqref{tL}, we must have
$[K_\kappa, K_\lambda]=K_{[\kappa,\lambda]_{\KY}}$\,, where we have denoted a `new' KY odd form by $[\kappa,\lambda]_{\KY}$. One can easily extract
an explicit expression for this form by noting that we must have $\{\ \ \}=X^b\hook [\kappa,\lambda]_{\KY}$.
By inverting this relation we find
\be\label{KY1}
[\kappa,\lambda]_{\KY}=\frac{1}{\pi}e_b\wedge \Big\{\kappa^a(\nabla_a\lambda^b)-\lambda^a(\nabla_a\kappa^b)+[\kappa^b, \tilde \lambda]+[\tilde \kappa,\lambda^b]\Bigr\}\,,
\ee
where quantities $\kappa^a\,, \tilde \kappa\,, \lambda^a\,, \tilde \lambda$ are related to $\kappa$ and $\lambda$ by \eqref{pomocne}.

The other two relations are quite analogous. In the case of $[K_\mu,M_\omega]$ we denote
\be\label{pomocne2}
\mu^a=X^a\hook \mu\,,\quad \tilde \mu=\frac{\pi-1}{2\pi}d\mu\,,\quad
\omega^a=e^a\wedge \omega\,,\quad \tilde \omega=-\frac{n-\pi-1}{2(n-\pi)}\delta \omega\,.
\ee
Then commutator $[K_\mu,M_\omega]$ is given by \eqref{pom} with $\kappa\leftrightarrow \mu$ and $\lambda\leftrightarrow \omega$ replaced. That is, we have
\be\label{pom2}
[K_\mu,M_\omega]=
[\mu^{(a},\omega^{b)}]\nabla_{a}\nabla_{b}+
\Big\{\mu^a(\nabla_a\omega^b)-\omega^a(\nabla_a\mu^b)+[\mu^b, \tilde \omega]-[\omega^b,\tilde \mu]\Bigr\}\nabla_b+\mbox{zeroth order}.
\ee
The requirement that the first term vanishes gives \eqref{fordcondKM}. The second term is now odd and so we must have $[K_\mu,M_\omega]=M_{[\mu,\omega]_{\KY}}$, where new even CCKY $[\mu,\omega]_{\KY}$ must obey
$\{\ \ \}=e^b\wedge [\mu,\omega]_{\KY}$.  By inverting this relation we obtain
\be\label{KY2}
[\mu,\omega]_{\KY}=\frac{1}{n-\pi}X_b\hook \Big\{\mu^a(\nabla_a\omega^b)-\omega^a(\nabla_a\mu^b)+[\mu^b, \tilde \omega]-[\omega^b,\tilde \mu]\Bigr\}\,,
\ee
where quantities $\mu^a\,, \tilde \mu\,, \omega^a\,, \tilde \omega$ are related to $\mu$ and $\omega$ by \eqref{pomocne2}.

Finally, in the third case we obtain
\be\label{KY3}
[\alpha,\beta]_{\KY}=\frac{1}{\pi}e_b\wedge\Big\{\alpha^a(\nabla_a\beta^b)-\beta^a(\nabla_a\alpha^b)+[\alpha^b, \tilde \beta]-[\beta^b,\tilde \alpha]\Bigr\}\,,
\ee
where
\be\label{pomocne3}
\alpha^a=e^a\wedge \alpha\,,\quad \tilde \alpha=-\frac{n-\pi-1}{2(n-\pi)}\delta \alpha\,,\quad
\beta^a=e^a\wedge \beta\,,\quad \tilde \beta=-\frac{n-\pi-1}{2(n-\pi)}\delta \beta\,,
\ee
which completes the proof of this proposition.
\end{proof}

%\subsection{Algebraic conditions}
Algebraic conditions \eqref{fordcond} are non-trivial and rather restrictive. By taking the trace with respect to indices $a$ and $b$ we obtain {\em necessary conditions}
\be\label{necesary_cond}
[X^{a}\hook\kappa,X_a\hook\lambda] = 0\;,\quad
    [X^{a}\hook\mu,e_a\wedge\omega] = 0\;,\quad
    [e^{a}\wedge\alpha,e_a\wedge\beta] = 0\;,
\ee
which shall be used in next subsections.
Using formulas for commutators listed in App.~\ref{apx:usefulid}, we can expand these conditions as follows
\begin{subequations}\label{fordcondctr}
\begin{gather}
  \sum_{k=1} \frac{(-1)^k}{(2k-1)!}\,\kappa\cwedge{2k}\lambda = 0
    \;,\label{fordcondctrKK}\\
  \sum_{k=0} \frac{(-1)^k}{(2k{+}1)!}\,\bigl((\pi{-}2k{-}1)\mu\bigr)\cwedge{2k{+}1}\omega = 0
    \;,\label{fordcondctrKM}\\
  \sum_{k=0} \frac{(-1)^k}{(2k{+}1)!}\,(\dm-\pi-2k-1)\,\bigl(\alpha\cwedge{2k{+}1}\beta \bigr)= 0
    \;.\label{fordcondctrMM}
\end{gather}
\end{subequations}
If the forms are homogeneous,  individual terms in the sums \eqref{fordcondctr} have different degrees and therefore have to vanish separately. The conditions then reduce to the sets of `individual' conditions
\begin{subequations}\label{fordcondhom}
\begin{align}
  &\kappa\cwedge{2k}\lambda = 0\;, &&k = 1,\,2, \dots
    \;,\label{fordcondhomKK}\\
  &\mu\cwedge{2k{+}1}\omega = 0\;, &&k = 0,\,1, \dots
    \;,\qquad 2k+1\neq r
    \;,\label{fordcondhomKM}\\
  &\alpha\cwedge{2k{+}1}\beta = 0\;, &&k = 0,\,1, \dots
    \;,\qquad 2k+1\neq p+q-\dm
    \;,\label{fordcondhomMM}
\end{align}
\end{subequations}
where ${p}$, ${q}$, and ${r}$ are degrees of forms ${\alpha}$, ${\beta}$, and ${\mu}$, respectively. Note that for ${\dm}$ even the exception ${2k+1 = p+q-\dm}$ cannot occur, so \eqref{fordcondhomMM} holds for all ${k}$.

\subsection{Lie derivative}\label{3.3}
Let us first consider the simplest possible case in which we have a KY 1-form ${\kappa}$ and an associated Killing vector ${k=\kappa^\sharp}$. In this case the operator ${K_\kappa}$ coincides with the Lie derivative ${\lied_{\!k}}$ acting on the Dirac bundle \cite{BennTucker:book}. The Lie derivative on the tangent bundle along a vector field ${k}$ is a generator of 1-parametric family of diffeomorfisms, induced from the manifold to the tangent bundle. Such diffeomorfisms can be lifted to the Dirac and Clifford bundles only if they conserve the metric structure, i.e., only if ${k}$ is a Killing vector. In such a case the action on the Dirac and Clifford bundle is given by\footnote{%
Recalling the standard action of the Lie derivative on the tangent bundle
${ \lied_{\!k} v^a = \nabla_{\!k}v^a - (\nabla_{\!b} k^a)v^b}$, one gets ${\lied_{\!k}\, e^a = \nabla_{\!k}\, e^a - (\nabla_{\!b}\,k^a)\, e^b +\frac14[d\kappa,e^a]=0}$, and similarly ${\lied_{\!k} X_a=0}$. It just states that the metric structures are conserved.}
\begin{equation}\label{liedDirClif}
    \lied_{\!k} \psi = \nabla_{\!k} \psi + \frac14 (d\kappa)\,\psi = K_\kappa\psi\;,\quad
    \lied_{\!k} \alpha = \nabla_{\!k} \alpha + \frac14 [d\kappa,\alpha]\;.
\end{equation}
\noindent This geometrical interpretation allows one to find a form of commutators if one of the operators is a Killing vector:
\begin{prop}[Lie derivative]\label{LD}
Let ${\kappa}$ be a KY 1-form, ${\lambda}$ an odd KY form, and ${\omega}$ an even CCKY form\footnote{In fact, this proposition can be generalized to arbitrary KY and CCKY forms $\lambda$ and $\omega$, see Prop.~\ref{Prop5.5}.}. Then the commutators of the corresponding operators
$K_\lambda$ and $M_\omega$ with ${K_\kappa}$ remain of the first-order and can be written as
\begin{equation}\label{Liedercom}
[K_\kappa,K_\lambda] = K_{\lied_k\lambda}\;,\qquad
    [K_\kappa,M_\omega] = M_{\lied_k\omega}\;,
\end{equation}
where ${k=\kappa^\sharp}$ is the Killing vector corresponding to ${\kappa}$.
In particular, if ${\lambda}$ is a KY 1-form corresponding to a Killing vector ${l=\lambda^\sharp}$, we have ${[K_\kappa,K_\lambda]=K_{[k,l]^\flat}}$, with ${[k,l]}$ being the Lie bracket.
As a corollary we have the following statement:
Let $k$ be a Killing vector, $\lambda$ a KY odd-form and $\omega$ a CCKY even-form, then
\be\label{new1}
\dot \lambda\equiv\lied_k\lambda\,,\quad
\dot \omega\equiv\lied_k\omega\,,
\ee
are `new' odd KY and even CCKY forms, respectively.
\end{prop}
\begin{proof}
To prove relations \eqref{Liedercom}, we first note that for a KY 1-form ${\kappa}$ conditions \eqref{fordcondKK} and \eqref{fordcondKM} are automatically satisfied ($\kappa^a={X^a\hook\kappa}$ is a 0-form). Using now Eq. \eqref{KY1} and recalling that $\lambda^a=X^a\hook \lambda$, we have
\ba
[\kappa,\lambda]_{\KY}&=&\frac{1}{\pi}e_b\wedge \Big\{\kappa^a(\nabla_a\lambda^b)-\lambda^a(\nabla_a\kappa^b)-\frac{1}{4}[\lambda^b,d\kappa]\Bigr\}\nonumber\\
&=&\kappa^a\nabla_a \lambda+\frac{1}{4}[d\kappa, \lambda]=\lied_k\lambda\,,
\ea
which proves the first statement. Similarly, using \eqref{KY2} and recalling that $\omega^a=e^a\wedge \omega$, we have
\ba
[\kappa,\omega]_{\KY}&=&\frac{1}{n-\pi}X_b\hook \Big\{\kappa^a(\nabla_a\omega^b)-\omega^a(\nabla_a\kappa^b)-\frac{1}{4}[\omega^b,d\kappa]\Bigr\}\nonumber\\
&=&\kappa^a\nabla_a \omega+\frac{1}{4}[d\kappa, \omega]=\lied_k\omega\,,
\ea
which completes the second relation \eqref{Liedercom}. The statement \eqref{new1} automatically follows.
\end{proof}

\subsection{Killing--Yano brackets}\label{ssc:KYbr}

In Prop.~\ref{pr:algfoop} we have defined an operation ${[\ ,\ ]_\KY}$ acting on odd KY and even CCKY forms, provided that conditions \eqref{fordcond} are satisfied.
In Sec.~\ref{3.3} we have found a form of these brackets in a special case. Namely, we have found that for a KY 1-form ${\kappa}$ and a general odd KY form ${\lambda}$ or a general even CCKY form $\omega$, the KY bracket reduces to the Lie derivative along the corresponding Killing vector $k=\kappa^\sharp$.

Now we find more explicit expressions for general KY brackets. The expressions \eqref{KY1}, \eqref{KY2}, and \eqref{KY3} can be written in terms of a {\em potential and co-potentials}, which is locally guaranteed by co-closeness of KY forms and closeness of CCKY forms.
\begin{prop}[Potentials and co-potentials]\label{Prop3.3}
Under the necessary conditions \eqref{necesary_cond}, the Killing--Yano brackets \eqref{KY1}, \eqref{KY2}, and \eqref{KY3} can be written in terms of potentials and co-potentials as follows
\begin{subequations}\label{KYbracketsPot}
\begin{align}
   [\kappa,\lambda]_\KY &=
     -\frac1\pi\,\delta\,\sum_{k=0}\frac{(-1)^k}{(2k{+}1)!}\,
     \bigl(\pi\kappa\bigr)\cwedge{2k}\bigl(\pi\lambda\bigr)\;,
   \label{KYbracketsPotKK}\\
   [\mu,\omega]_\KY &=
     \frac1{\dm-\pi}\,d\, \sum_{k=0}\frac{(-1)^k}{(2k{+}1)!}\,
     %\Bigl(\frac\pi{\pi{-}2k}\,\mu\Bigr)
     \Bigl((\pi{-}2k)^{-1}\pi\,\mu\Bigr)
     \cwedge{2k+1}\Bigl((\dm{-}\pi)\omega\Bigr)\;,
   \label{KYbracketsPotKM}\\
   [\alpha,\beta]_\KY &=
     \frac1\pi\,\delta\,\sum_{k=0}\frac{(-1)^k}{(2k{+}1)!}\,
     \frac1{\dm{-}\pi{-}2k}\,\Bigl((\dm{-}\pi)\alpha\Bigr)
     \cwedge{2k+1}\Bigl((\dm{-}\pi)\beta\Bigr)\;.
    \label{KYbracketsPotMM}
\end{align}
\end{subequations}
The expression \eqref{KYbracketsPotMM} is valid if the factor ${(\dm{-}\pi{-}2k)}$ in the denominator cannot become zero. It includes the case of an odd dimension, or the case of an even dimension and forms ${\alpha,\,\beta}$ of maximal degrees ${p,\,q}$ such that ${p+q\leq\dm}$.
\end{prop}
\begin{proof}
The proof of \eqref{KYbracketsPotKK} starts with the expression \eqref{KY1} which can be rewritten as
\begin{equation}\label{KYpotproof1}
  \pi[\kappa,\lambda]_\KY = \,e_b\wedge\Bigl(
    {\textstyle\frac12}\bigl[\nabla^b\kappa^a,\lambda_a\bigr]_+\!
    -{\textstyle\frac12}\bigl[\kappa^a,\nabla^b\lambda_a\bigr]_+\!
    +\bigl[\kappa^b,\tilde\lambda\bigr]+\bigl[\tilde\kappa,\lambda^b\bigr]\Bigr)
    - {\textstyle\frac12}\, d \bigl[\kappa^a,\lambda_a\bigr] \;.
\end{equation}
Here, we have used the antisymmetry ${\nabla_a\kappa_b=-\nabla_b\kappa_a}$ which follows from the KY property of ${\kappa}$, and rewritten the terms ${\kappa^a(\nabla_b\lambda_a)}$ and ${\lambda^a(\nabla_b\kappa_a)}$ using commutators and anticommutators. The last term in \eqref{KYpotproof1} vanishes thanks to necessary conditions \eqref{necesary_cond}. The (anti)-commutators can be expanded using relations \eqref{useful1_odd}. Next, we transform the contraction over index ${a}$ into contracted wedge and use the KY equation. Substituting \eqref{pomocne} for ${\tilde\kappa}$ and ${\tilde\lambda}$ than leads to
\begin{equation}\label{KYpotproof2}
\begin{split}
  \pi[\kappa,\lambda]_\KY = \sum_{k=1}\frac{(-1)^k}{(2k{-}1)!}\,
     \,e^b\wedge\Bigl(
     &(2k{-}1)\kappa\cwedge{2k{-}1}\bigl(X_b\hook{\textstyle\frac1\pi}d\lambda\bigr)
     -(2k{-}1)\bigl(X_b\hook{\textstyle\frac1\pi}d\kappa\bigr)\cwedge{2k{-}1}\lambda\\
     &+\bigl(X_b\hook\kappa\bigr)\cwedge{2k{-}1}\bigl({\textstyle\frac{\pi{-}1}\pi}d\lambda\bigr)
     +\bigl({\textstyle\frac{\pi{-}1}\pi}d\kappa\bigr)\cwedge{2k{-}1}\bigl(X_b\hook\lambda\bigr)
     \Bigr)\;.
\end{split}
\end{equation}
Applying relations \eqref{mywedge_hookWedge} and some algebra gives
\begin{equation}\label{KYpotproof3}
\begin{split}
  \pi[\kappa,\lambda]_\KY = &\sum_{k=1}\frac{(-1)^k}{(2k{-}1)!}\,\Bigl(
     -\bigl(\pi\kappa\bigr)\cwedge{2k{-}1}\bigl({\textstyle\frac{\pi{-}1}\pi}d\lambda\bigr)
     +\bigl({\textstyle\frac{\pi{-}1}\pi}d\kappa\bigr)\cwedge{2k{-}1}\bigl(\pi\lambda\bigr)\Bigr)\\
     &+\sum_{k=1}\frac{(-1)^{k+1}}{(2k{-}1)!}\,
     \Bigl(\kappa\cwedge{2k{+}1}\bigl({\textstyle\frac1\pi}d\lambda\bigr)
     -\bigl({\textstyle\frac1\pi}d\kappa\bigr)\cwedge{2k{+}1}\lambda\Bigr)\;.
\end{split}
\end{equation}
Using expansions \eqref{useful1_odd}, relations \eqref{A24}, ${\delta\kappa=\delta\lambda=0}$, the KY property of ${\kappa}$ and ${\lambda}$, and definition \eqref{cwedgedef}, it can be shown that the second sum is equal to ${\frac12\delta[\kappa^a,\lambda_a]}$. Therefore, it vanishes thanks to conditions \eqref{necesary_cond}. By a similar calculation, the first sum is equal to
\begin{equation}\label{KYpotproof4}
  \pi[\kappa,\lambda]_\KY = -\,\delta\,\sum_{k=0}\frac{(-1)^k}{(2k{+}1)!}\,
    \bigl(\pi\kappa\bigr)\cwedge{2k}\bigl(\pi\lambda\bigr)\;,
\end{equation}
which is the desired relation \eqref{KYbracketsPotKK}.

The proofs of relations \eqref{KYbracketsPotKM} and \eqref{KYbracketsPotMM} are analogous. However, in the case of two even CCKY forms, the expression \eqref{KYbracketsPotMM} is correct only if it is well defined. There is a potential problem with the vanishing denominator ${(\dm{-}\pi{-}2k)}$. Fortunately, in some important cases as for an odd dimension and even forms ${\alpha}$, ${\beta}$, this cannot happen.

Terms for which the denominator vanishes must be treated separately. It cannot be done in a simple way for general inhomogeneous forms. For such forms it is better to use Hodge duality which will be discussed in Sec.~\ref{ssc:HodgeDual}. However for homogeneous forms ${\alpha}$, ${\beta}$ one can identify the problematic terms. The result is presented in the following proposition.

Note that the prefactor ${\frac1\pi}$ in \eqref{KYbracketsPotKK} and \eqref{KYbracketsPotMM} cannot be divergent, since the result of the KY bracket is odd, i.e., it has a non-zero degree. The factor ${\frac1{n-\pi}}$ in \eqref{KYbracketsPotKM} is not divergent, because for ${[\mu,\omega]_\KY}$ of degree ${\dm}$ the form ${\omega}$ must be also of degree ${\dm}$ and corresponding operators ${M_{[\mu,\omega]_\KY}}$ and ${M_{\omega}}$ would be vanishing.
\end{proof}

Expressions \eqref{KYbracketsPot} simplify significantly under the assumption that all the KY and CCKY forms are homogeneous. In that case we can use \eqref{fordcondhom} to find:
\begin{prop}[Killing--Yano brackets for homogeneous forms]\label{PropKYbrhom}
For odd homogeneous forms ${\kappa}$, ${\lambda}$, ${\mu}$ of degrees ${p}$, ${q}$, ${r}$, respectively, and even homogeneous forms ${\alpha}$, ${\beta}$, ${\omega}$ of degrees ${a}$, ${b}$, and ${c}$, respectively, the Killing--Yano brackets take the form:
\begin{subequations}\label{KYbracketsHom}
\begin{align}
   [\kappa,\lambda]_\KY &=
     -\frac{p\,q}{p{+}q{-}1}\,\delta
     \bigl(\kappa\wedge\lambda\bigr) \;,
      \label{KYbracketsHomKK}\\[1ex]
   [\mu,\omega]_\KY &=
     \frac{r\,(\dm{-}c)}{\dm{+}r{-}c{-}1}
     \frac{(-1)^{\frac{r{-}1}{2}}}{r!}\,
     d\bigl(\mu\cwedge{r}\omega\bigr)\;,
   \label{KYbracketsHomKM}\\[1ex]
   [\alpha,\beta]_\KY &=
   \begin{cases}
     \displaystyle
     \frac{(\dm{-}a)(\dm{-}b)}{2\dm{-}a{-}b{-}1}
     \frac{(-1)^{\frac{a{+}b{-}\dm{-}1}{2}}}{(a{+}b{-}\dm)!}\,
     \delta\bigl(\alpha\cwedge{p{+}q{-}\dm}\beta\bigr)
       &\mspace{-255mu}\text{for ${\dm}$ odd}\;,\\[1.5ex]
     \displaystyle
     0\quad
       &\mspace{-255mu}\text{for ${\dm}$ even and ${a+b\leq\dm}$}\;,\\[1ex]
     \displaystyle
     \frac{(\dm{-}a)(\dm{-}b)}{2\dm{-}a{-}b{+}1}\frac1{(a{+}b{-}\dm{-}1)!}
       \Bigl(\alpha\cwedge{a{+}b{-}\dm{-}1}\bigl({\textstyle\frac1{n{-}b{+}1}}\delta\beta\bigr)
       -\bigl({\textstyle\frac1{n{-}a{+}1}}\delta\alpha\bigr)\cwedge{a{+}b{-}\dm{-}1}\beta\Bigr)&\\
       &\mspace{-255mu}\text{for ${\dm}$ even and ${a+b>\dm}$}\;.
   \end{cases}\label{KYbracketsHomMM}\raisetag{14ex}
   \end{align}
\end{subequations}
\end{prop}
\noindent {\em Remark:} We would like to emphasize here that relations \eqref{KYbracketsPot} and \eqref{KYbracketsHom} have been derived using the necessary conditions \eqref{necesary_cond}. By construction, however, one can guarantee that the Killing--Yano brackets produce KY or CCKY forms only under the more restrictive conditions \eqref{fordcond}.
At present (and apart from the case of a Killing 1-form discussed in Sec.~\ref{3.3}) we do not know any example for which conditions \eqref{fordcond} are satisfied and the Killing--Yano brackets are non-vanishing. One might speculate that these conditions are in fact so strong that they guarantee that the Killing--Yano brackets (of homogeneous forms) are automatically trivial.
To support this idea, simple examples in which conditions \eqref{fordcond} are satisfied and Killing--Yano brackets vanish can be easily found in flat or pp-wave spacetimes. Another, highly non-trivial example illustrating this point is discussed in the next section.

%%%%%%%%%%%%%%%%%%%%%%%%%%%%%%%%%%%%%%%%%%%%%%%%%%%%%%%%%%%%%%%%%%%%%%%%%%
%%%%%%%%%%%%%%%%%%%%%%%%%%%%%%%%%%%%%%%%%%%%%%%%%%%%%%%%%%%%%%%%%%%%%%%%%%

\section{PCKY tensor and complete set of commuting operators}\label{sc:PCKYsym}
In this section we shall study commutators of the Dirac symmetry operators generated from the principal conformal Killing--Yano tensor \cite{KrtousEtal:2007jhep}. We demonstrate that among these operators one can find a complete subset of mutually commuting operators. These operators are related to separability  \cite{OotaYasui:2008} of the Dirac equation in Kerr-NUT-(A)dS spacetimes in all dimensions  \cite{CKK:separability}. We closely follow \cite{FrolovKubiznak:2008} to review the definition of the principal conformal Killing--Yano tensor and the procedure to generate from it the tower of explicit and hidden symmetries.

\subsection{Principal conformal Killing--Yano tensor}\label{PCKYsec}
The {\em principal conformal Killing--Yano (PCKY)} tensor ${h}$ is a non-degenerate CCKY 2-form
\cite{KrtousEtal:2007jhep}.
This means that there exists a 1-form ${\xi}$ so that
\be\label{PCKY}
\nabla_{a}{h}={e_a}\wedge {\xi}\, ,\quad {\xi}=-\frac{1}{n-1}{\delta h}\,.
\ee
{\em Non-degeneracy} requires that in a generic point of the manifold the skew symmetric matrix $h_{ab}$
has the maximum possible (matrix) rank and that the eigenvalues of ${h}$ are
functionally independent in some spacetime domain. Eq. \eqref{PCKY} automatically implies $dh=0$ and
hence there exists a 1-form, a PCKY potential, so that
\be
{h}={db}\,.
\ee
The 1-form ${\xi}$ associated with ${h}$ is called {\em primary} and it is a Killing 1-form \cite{KrtousEtal:2008}.

The most general metric, a {\em canonical metric}, admitting the PCKY tensor was constructed
in \cite{HouriEtal:2007, KrtousEtal:2008}.
If we parametrize $n=2N+\eps\,,$ $N=[n/2]$, the metric and the PCKY potential are
\ba
{g}_{\mbox{\tiny can}}\!&=&\!\sum_{\mu=1}^N\Bigl[\frac{U_\mu}{X_\mu}{d}x_{\mu}^2
  +\frac{X_\mu}{U_\mu}\Bigl(\sum_{j=0}^{N-1} A_{\mu}^{(j)}{d}\psi_j\!\Bigr)^{\!2}\Bigr]+
\eps S\Bigl(\sum_{j=0}^{N} A^{(j)}{d}\psi_j\Bigr)^2\,, \ \label{KNdS2}\\
{b}\!&=&\!\frac{1}{2}\sum_{k=0}^{N-1} A^{(k+1)}{d}\psi_k\,,\label{b}
\ea
where
\begin{gather}
U_{\mu}=\prod_{\nu\ne\mu}(x_{\nu}^2-x_{\mu}^2)\,,\quad S=\frac{-c}{A^{(N)}}\,,\nonumber\\
A^{(k)}_\mu=\!\!\!\!\!\sum_{\substack{\nu_1<\dots<\nu_k\\\nu_i\ne\mu}}\!\!\!\!\!x^2_{\nu_1}\dots x^2_{\nu_k},\quad
A^{(k)}=\!\!\!\!\!\sum_{\nu_1<\dots<\nu_k}\!\!\!\!\!x^2_{\nu_1}\dots x^2_{\nu_k}\;\label{vztahy},
\end{gather}
the quantities $X_\mu=X_\mu(x_\mu)$ are arbitrary functions of one variable which we call {\em metric functions},
and $c$ is a constant. Introducing the basis
\be\label{omega}
{E}^{\mu} =\sqrt{\frac{U_\mu}{X_\mu}}\,
{d}x_\mu\,,\quad
{\tilde  E}^{\mu} = \sqrt{\frac{X_\mu}{U_\mu}}
 \sum_{j=0}^{N-1}A_{\mu}^{(j)}{d}\psi_j\;,\quad
 {E}^{0} =\sqrt{S}\sum_{j=0}^{N} A^{(j)}{d}\psi_j\,,
\ee
the metric and the PCKY tensor take the form
\ba
{g}_{\mbox{\tiny can}}\!\!\!&=&\!\sum_{\mu=1}^N\bigl({E}^{\mu}\otimes{E}^{\mu}+
{\tilde E}^{\mu}\otimes{\tilde E}^{\mu}\bigr)+\eps E^0\otimes E^0\,,\label{gbasis}\\
{h}\!&=&\!{db}=\sum_{\mu=1}^N x_\mu\, {E}^{\mu}
\wedge \!{\tilde E}^{\mu}\,.\label{PCKY_2n}
\ea
This means that the chosen basis is the Darboux basis and that coordinates $x_\mu$
are natural coordinates associated with the `eigenvalues' of the PCKY tensor.

\subsection{Towers of symmetries}
The PCKY tensor generates a tower of even CCKY tensors
\be\label{hj}
{h}^{(j)}\equiv {h}^{\wedge j}=\underbrace{{h}\wedge \ldots \wedge
{h}}_{\mbox{\tiny{total of $j$ factors}}}\, .
\ee
${h}^{(j)}$ being a $(2j)$-form. In particular ${h^{(0)}=1}$ and ${h}^{(1)}={h}$.
Since ${h}$ is non-degenerate, one has a
set of $N$ non-vanishing CCKY forms. In an even dimensional
spacetime  ${h}^{(N)}$ is proportional to the totally antisymmetric
tensor, whereas it is dual to a Killing vector in odd dimensions. In both cases
such a CCKY tensor is trivial and can be excluded from the tower of hidden symmetries.
Therefore, we take $j=1,\dots, N-1$.
The CCKY tensors \eqref{hj} can be generated from the potentials ${b}^{(j)}$
\be\label{bj}
{b}^{(j)}\equiv{b}\wedge {h}^{\wedge (j-1)}\,,\quad
{h}^{(j)}={d}{b}^{(j)}\,.
\ee
Each $(2j)$-form ${h}^{(j)}$
determines a $(n-2j)$-form of the KY tensor
\be\label{fj}
{f}^{(j)}\equiv{*}{h}^{(j)}\, .
\ee
In their turn, these tensors give rise to the Killing tensors
${K}^{(j)}$
\be\label{Kj}
K^{(j)}_{ab}\equiv \frac{1}{(n-2j-1)!(j!)^2} f^{(j)}_{\, \, \, \, \, a c_1\ldots c_{n-2j-1}}
f_{\, \, \, \, \, b}^{(j) \, \,  c_1\ldots c_{n-2j-1}}\, .
\ee
The coefficient in this definition is chosen so that we get
\be
{K}^{(j)}=\sum_{\mu=1}^N A^{(j)}_\mu\bigl({E}^{\mu}\otimes{E}^{\mu}+
{\tilde E}^{\mu}\otimes{\tilde E}^{\mu}\bigr)+\eps A^{(j)}E^0\otimes E^0\,.
\ee
We also define ${K}^{(0)}=g$ [cf. Eq. \eqref{gbasis}]. Hence we can take the range of $j$ as $j=0,\dots N-1$.

The PCKY tensor also naturally generates $(N+\eps)$ 2-forms $\omega^{(k)}$ $(k=0,\dots,N-1+\eps)$ which are {\em Killing co-potentials} for the
Killing vectors $\partial_{\psi_k}$ \cite{CveticEtal:2010}.
These are given by
\be
\omega^{(j)}_{ab}=\frac{1}{n-2j-1}K^{(j)}_{ac}h^c_{\ b}\,,\quad \omega^{(N)}=\frac{\sqrt{-c}}{N!}*b^{(N)}\,.
\ee
In the canonical basis we have [cf. Eq. \eqref{PCKY_2n}]
\be
\omega^{(j)}=\frac{1}{n-2j-1}\sum_{\mu=1}^N x_\mu A_\mu^{(j)} E^\mu\wedge \tilde E^\mu\,.
\ee
[Note that $\omega^{(0)}=h/(n-1)$.] One then has\footnote{%
Note that although the gauge freedom $b\to b+d\lambda$ affects $\omega^{(N)}$,
$\omega^{(N)}\to \omega^{(N)}+\frac{1}{N!}*d(\lambda h^{(N-1)})\,,$
$\delta \omega^{(N)}$ remains unchanged.
}
\ba\label{KVs}
\xi^{(j)}=-\delta\omega^{(j)}=K^{(j)}\cdot \xi=(\partial_{\psi_j})^\flat\,,\quad
\xi^{(N)}=-\delta\omega^{(N)}={f}^{(N)}=(\partial_{\psi_N})^\flat\,.
\ea
The third equality in these expressions was first established in \cite{KrtousEtal:2007jhep}.

The constructed hidden symmetries obey for any $k$ and $l$ the following relations:
\be\label{Lie_can}
\lied_{{\xi}^{(k)\,\sharp}} {\xi}^{(l)}=0\,,\quad
\lied_{{\xi}^{(k)\,\sharp}} h^{(l)}=0\,,\quad
\lied_{{\xi}^{(k)\,\sharp}} f^{(l)}=0\,.
\ee
which can be rewritten in terms of Killing--Yano brackets as
\begin{equation}\label{KYxihf=0}
[{\xi}^{(k)}, {\xi}^{(l)}]_\KY = 0 \,,\quad
[{\xi}^{(k)}, h^{(l)}]_\KY = 0 \,,\quad
[{\xi}^{(k)}, f^{(l)}]_\KY=0\,.
\end{equation}
These Killing--Yano brackets are well-defined since the necessary conditions \eqref{fordcond} are for any KY 1-form ${\xi^{(k)}}$ automatically satisfied. It is natural to ask what are mutual Killing--Yano brackets of the remaining symmetries ${h^{(l)}}$ and $f^{(k)}$.
\begin{prop}[Killing--Yano brackets and towers of hidden symmetries]\label{ttt}
The towers of hidden symmetries $\{h^{(l)}, f^{(k)}\}$, generated from the PCKY tensor by Eq. \eq{hj} and \eq{fj} satisfy the following relations:
\begin{align}
[e_{(a}\wedge h^{(k)}, e_{b)}\wedge h^{(l)}]&=0 \,,\qquad [h^{(k)}, h^{(l)}]_{\KY}=0\,,\label{sat1}\\
[X_{(a}\hook f^{(k)}, X_{b)}\hook f^{(l)}]&=0 \,,\qquad  [f^{(k)}, f^{(l)}]_\KY=0\,. \label{sat2}
\end{align}
In odd number of spacetime dimensions one also has
\begin{equation}\label{sat3}
[X_{(a}\hook f^{(k)}, e_{b)}\wedge h^{(l)}]=0\,,\qquad [f^{(k)}, h^{(l)}]_\KY=0\,.
\end{equation}
\end{prop}
\begin{proof}
The first relation is proved in App.~\ref{apx:PCKYsym}, (Lemmas~\ref{lemma_tower1}, \ref{lemma_tower2}, and \ref{lemma_tower3}). Other relations follow
by employing Hodge duality and Prop.~\ref{pr:HDKYbr}.%, and Prop.~\ref{Prop3.4}.
\end{proof}
\noindent In odd dimensions \eqref{sat2} can be rephrased using Prop.~\ref{Prop3.4} as an interesting corollary about the Schouten--Nijenhuis brackets (cf.\ Sec.~\ref{ssc:SNbr}) of KY forms ${f^{(k)}}$
\begin{equation}\label{SNff=0}
   [f^{(k)}, f^{(l)}]_\SN=0\,.
\end{equation}
\noindent {\em Remark:} the results of prop.~\ref{ttt} also apply to the tower of symmetries built from a degenerate CCKY 2--form. For a CCKY 2--form with some degenerate eigenvalues the metric is not the canonical metric discussed in sec.~\ref{PCKYsec} and not all the operators in the tower are independent. Prop.~\ref{ttt} in any case guarantees that all operators in the tower commute among each other. These metrics have been classified in \cite{HouriEtal:2008b}.

\subsection{Complete set of commuting operators}\label{apx:setcomop}

We have now all the tools necessary to show that the canonical spacetime described above admits a complete set of mutually commuting symmetry operators for the Dirac equation.
\begin{prop}[Complete set of commuting operators]\label{CSOCO}
The most general spacetime admitting the PCKY tensor admits the following complete set of commuting operators:
\be\label{symop}
\{D, K_{\xi^{(0)}},\dots K_{\xi^{{(N-1+\eps)}}}, M_{h^{(1)}},\dots M_{h^{(N-1)}}\}\,.
\ee
Here, $D$ is the Dirac operator, $K_{\xi^{(k)}}$ are operators \eqref{Kdef} corresponding to Killing forms $\xi^{(k)}$,
\eqref{KVs}, and $M_{h^{(i)}}$ are operators \eqref{Mdef} connected with even CCKY forms  $h^{(i)}$, \eqref{hj}.
\end{prop}
\noindent {\em Remark:} Note that the Dirac operator can be
    written as ${D=M_{h^{(0)}}}$.
\begin{proof}[Proof]
From Prop.~\ref{2.2} all operators commute with the Dirac operator $D$. Next, Prop.~\ref{pr:algfoop} together with Eq.~\eqref{KYxihf=0} guarantee that all operators $K$ commute between each other and with operators $M$.
Lastly, Prop.~\ref{pr:algfoop} together with \eqref{sat1} guarantee that operators $M$ mutually commute among themselves.
The `independence' of these operators
follows from the independence of corresponding CKY forms.
\end{proof}
\noindent
In even dimensions the set \eqref{symop} exhausts all possibilities (there are no further first order operators commuting with the Dirac operator available).
However, in odd dimensions one also has operators commuting with the Dirac operator corresponding to odd KY forms $f^{(j)}$, \eqref{fj}. Therefore, in odd dimensions one has a choice of different complete sets of commuting operators; instead of each $M_{h^{(i)}}$ one can take $K_{*h^{(i)}}$.
Prop.~\ref{pr:algfoop}, together with Eqs.~\eqref{sat2},~\eqref{sat3} and the independence of corresponding tensors, guarantees that we have another complete set of commuting operators. In particular one has the the following obvious choice:
\begin{prop}[Another complete set in odd dimensions]
In odd dimensions the canonical spacetimes admit another complete set of commuting operators, given by
\be\label{symop2}
\{D, K_{\xi^{(0)}},\dots K_{\xi^{{(N-1+\eps)}}}, K_{f^{(1)}},\dots K_{f^{(N-1)}}\}\,.
\ee
\end{prop}

\subsection{Separability in Kerr-NUT-(A)dS spacetimes}
The canonical spacetime described in Sec.~\ref{PCKYsec} captures the most general known Kerr-NUT-AdS spacetimes
in all dimensions \cite{ChenEtal:2006cqg}. These spacetimes are recovered
 when metric functions take the form
\be\label{X}
X_{\mu}=\sum\limits_{k=\varepsilon}^{N}c_kx_{\mu}^{2k}-2b_{\mu}x_{\mu}^{1-\varepsilon}+\frac{\varepsilon c}{x_{\mu}^2}\,.
\ee
Here, $c_N$ is proportional to the cosmological constant and the other constants are related to the mass, NUT-parameters, and rotations.\footnote{%
In fact, Kerr-NUT-(A)dS spacetimes assume the canonical form \eqref{KNdS2} only if all rotation parameters are nonzero and distinct.
In the opposite case the CCKY tensor $h$ is degenerate and the metric is slightly different, see \cite{HouriEtal:2008b}.}

Hidden symmetries in Kerr-NUT-AdS spacetimes allow for separating variables in the Hamilton--Jacobi \cite{FrolovEtal:2007}, scalar \cite{FrolovEtal:2007},
and Dirac \cite{OotaYasui:2008} equations in all dimensions. In the scalar case the link between symmetry operators and the separability
has been already established  \cite{SergyeyevKrtous:2008, FrolovKrtous:2011}; the result for the Hamilton--Jacobi equation follows by geometric optics approximation
\cite{SergyeyevKrtous:2008}. What remains is to intrinsically characterize the separability of the Dirac equation.

It follows from Prop.~\ref{CSOCO} that Kerr-NUT-(A)dS spacetimes in all dimensions possess a complete set of mutually commuting operators.
This set is given by the operators \eqref{symop} in even dimensions, whereas in  odd dimensions one can choose any complete subset of operators \eqref{symop} and \eqref{symop2}.
It is an interesting question to ask if a proper choice of these operators underlies separability of the massive Dirac equation demonstrated by Oota and Yasui \cite{OotaYasui:2008}.
In fact, it was demonstrated by Carter and McLenghan \cite{CarterMcLenaghan:1979}  that the Chandrasekhar's separation constants in 4D \cite{Chandrasekhar:1976} are the eigenvalues of operators \eqref{symop}. In 5D Wu found a separated solution for the massive Dirac equation in the Kerr-(A)dS spacetime which is an eigenfuction of operators \eqref{symop2} \cite{Wu:2008b}.
In all dimensions, the link between the symmetry operators \eqref{symop} or \eqref{symop2} and the separated solution of Oota and Yasui \cite{OotaYasui:2008} will be shown in \cite{CKK:separability}.

%%%%%%%%%%%%%%%%%%%%%%%%%%%%%%%%%%%%%%%%%%%%%%%%%%%%%%%%%%%%%%%%%%%%%%%%%%
%%%%%%%%%%%%%%%%%%%%%%%%%%%%%%%%%%%%%%%%%%%%%%%%%%%%%%%%%%%%%%%%%%%%%%%%%%

\section{First order symmetries of the Dirac operator} \label{sc:1ordsym}
In Sec.~\ref{sc:comopKYbr} we have focused our attention on operators {\em commuting} with the Dirac operator. Such operators play a central role for separability of the Dirac equation. However, commuting operators represent only a subset of all symmetry operators of the Dirac operator described by Prop.~\ref{1.1}.
One can consider operators that anti-commute with $D$, as well as other symmetry operators. Especially useful are \emph{graded symmetry operators (GSOs)}, i.e., operators that graded commute or anti-commute with $D$.
%They naturally appear in the context of SUSY; we shall use them in Sec.~\ref{SUSY} to discuss a relationship between the operator level and the  classical spinning particle description.
More general symmetry operators also allow us to study more general Killing--Yano brackets and extend the validity of some of the theorems of Sec.~\ref{sc:comopKYbr}. The aim of this section is not to deliver an exhaustive treatment of the subject but rather to provide a framework which can be worked out more thoroughly if needed.

\subsection{Killing--Yano brackets for general KY forms}\label{ssc:genKYbr}
We start with the anti-commuting version of Prop.~\ref{2.2}. By similar arguments as in Sec.~\ref{ssc:commuting} one finds that:
\begin{prop}[Anticommutation with ${D}$]
The most general first order operator ${S}$ anticommuting with ${D}$ is
\begin{equation}\label{anticomwD}
\begin{aligned}
    &S = K_{\omega_\even} + M_{\omega_\odd}\;,\\
    &K_{\omega_\even} = X^a\hook\omega_\even\nabla_{\!a}
       + \frac{\pi-1}{2\pi}d\omega_\even\;,
       \quad&&\text{with ${\omega_\even}$ being an even KY form}\;,\\
    &M_{\omega_\odd} = e^a\wedge\omega_\odd\nabla_{\!a}
       - \frac{\dm-\pi-1}{2(\dm-\pi)}\delta\omega_\odd\;,
       \quad&&\text{with ${\omega_\odd}$ being an odd CCKY form}\;.
\end{aligned}
\end{equation}
\end{prop}
\noindent This allows us to extend the definition of Killing--Yano brackets to general (not necessarily odd) KY forms:
\begin{prop}[KY bracket for KY forms]
Let ${\kappa}$, ${\lambda}$ be two KY forms. Sufficient and necessary conditions that the commutator ${[K_\kappa,K_\lambda]}$ remains of the first order are
\begin{equation}\label{algcondgenKYbr}
    [X^{(a}\hook\kappa,X^{b)}\hook\lambda]=0\;.
\end{equation}
If these conditions are satisfied, we can define the Killing--Yano brackets ${[\kappa,\lambda]_\KY}$ as
\begin{equation}\label{genKYbr}
    [K_\kappa,K_\lambda]=K_{[\kappa,\lambda]_\KY}\;.
\end{equation}
The conditions \eqref{algcondgenKYbr} are equivalent to
\begin{equation}\label{algcondgenKYbrexpl}
\begin{aligned}
    \sum_{k=1}\frac{(-1)^k}{(2k{-}1)!}\,\bigl(X^{(a}\hook\kappa\bigr)\cwedge{2k{-}1}\bigl(X^{b)}\hook\lambda\bigr)&=0
      &\quad&\text{for at least one form odd}\,,\\
    \sum_{k=0}\frac{(-1)^k}{(2k)!}\,\bigl(X^{(a}\hook\kappa\bigr)\cwedge{2k}\bigl(X^{b)}\hook\lambda\bigr)&=0
      &\quad&\text{for both forms even}\,.
\end{aligned}
\end{equation}
For homogeneous forms each of terms in the sums must vanish separately.
\end{prop}
\begin{proof}
The commutator ${[K_\kappa,K_\lambda]}$ is again given by Eq.~\eqref{pom}. The second order term gives the conditions \eqref{algcondgenKYbr} which transfers to \eqref{algcondgenKYbrexpl} using expansions \eqref{useful1_odd} and \eqref{useful1_even}. For homogeneous forms we realize that the different terms in the sum are of the different degree.

It remains to check that the commutator reduces to ${K_\varphi}$ for some KY form ${\varphi}$. If both ${\kappa}$ and ${\lambda}$ are even, the operators ${K_\kappa}$ and ${K_\lambda}$ are odd and they anticommute with ${D}$. Hence, their commutator is even and commutes with ${D}$. According to Prop.~\ref{2.2} it is given by ${K_\varphi}$ with an odd KY form ${\varphi}$. If instead ${\kappa}$ is odd and ${\lambda}$ even, ${K_\kappa}$ is even and commuting with ${D}$ and ${K_\lambda}$ is odd and anticommuting with ${D}$. Their commutator is odd and anticommuting with ${D}$, hence, thanks to Prop.~\ref{anticomwD}, it is given again by ${K_\varphi}$ with ${\varphi}$ being an even KY form.
\end{proof}

Let us mention that the contraction of conditions \eqref{algcondgenKYbr} with respect to the indices ${a}$ and ${b}$ leads to simpler necessary conditions, which for homogeneous forms give a set of relations
\begin{equation}\label{ctrcondgenhomKYbr}
\begin{aligned}
    \kappa\cwedge{2k}\lambda &= 0\;, &&k=1,2,\dots &\quad&\text{for at least one of the forms odd}\,,\\
    \kappa\cwedge{2k-1}\lambda &= 0\;, &&k=1,2,\dots &\quad&\text{for both forms even}\,.
\end{aligned}
\end{equation}

By the same arguments as in Sec.~\ref{ssc:KYbr} the expression \eqref{KYbracketsPotKK} can be extended to general KY forms, namely:
\begin{equation}\label{genKYbracketsPot}
\begin{aligned}
   &[\kappa,\lambda]_\KY =
     -\frac1\pi\,\delta\,\sum_{k=0}\frac{(-1)^k}{(2k{+}1)!}\,
     \bigl(\pi\kappa\bigr)\cwedge{2k}\bigl(\pi\lambda\bigr)&\quad
     &\text{for ${\kappa}$ odd}\;,\\
   &\begin{split}
   [\kappa,\lambda]_\KY &=
     \frac1\pi\Bigl(\bigl(\pi\kappa\bigr)\wedge\bigl({\textstyle\frac{\pi-1}\pi}d\lambda\bigr)
     -\bigl({\textstyle\frac{\pi-1}\pi}d\kappa\bigr)\wedge\bigl(\pi\lambda\bigr)\Bigr)\\
     &\qquad
     -\frac1\pi\,\delta\,\sum_{k=1}\frac{(-1)^k}{(2k)!}\,
     \bigl(\pi\kappa\bigr)\cwedge{2k{-}1}\bigl(\pi\lambda\bigr)
   \end{split}&\quad
     &\text{for ${\kappa}$, ${\lambda}$ even}\;.
\end{aligned}
\end{equation}
For homogeneous forms ${\kappa}$, ${\lambda}$ of degrees ${p}$, ${q}$, respectively, we can employ the conditions \eqref{ctrcondgenhomKYbr} and the KY bracket simplifies to:
\begin{equation}\label{genhomKYbracketsPot}
   [\kappa,\lambda]_\KY =
   \begin{cases}
     \displaystyle
     -\frac{pq}{p+q-1}\,\delta\,\bigl(\kappa\wedge\lambda\bigr)
     &\text{for ${\kappa}$ odd}\;,\\[2ex]
     \displaystyle
     \frac{pq}{p+q+1}\,\Bigl(\frac1{q+1}\,\kappa\wedge d\lambda
     -\frac1{q+1}\, d\kappa\wedge \lambda\Bigr)\quad
     &\text{for ${\kappa}$, ${\lambda}$ even}\;.
   \end{cases}
\end{equation}

A similar procedure could be applied to commutators of arbitrary KY and CCKY forms. We will not do that in details. Let us mention only the generalization of Prop.~\ref{LD}:
\begin{prop}[Lie derivative]\label{Prop5.5}
Let ${\kappa}$ be a KY 1-form, ${\lambda}$ an arbitrary (not necessary odd) KY form, and ${\omega}$ an arbitrary (not necessary even) CCKY form. Then the commutators of the corresponding operators ${K_\lambda}$ and ${M_\omega}$ with ${K_\kappa}$ remain of the first-order and they can be written as
\begin{equation}
  [K_\kappa,K_\lambda] = K_{\lied_k\lambda}\;,\quad
  [K_\kappa,M_\omega]  = M_{\lied_\kappa\omega}\;,
\end{equation}
where ${k={}^\sharp\kappa}$ is the Killing vector associated with ${\kappa}$.
As a corollary we have the following statement: Let $k$ be a Killing vector, $\lambda$ an arbitrary KY form and $\omega$ an arbitrary CCKY form,
then
\be\label{new2}
\dot \lambda\equiv\lied_k\lambda\,,\quad
\dot \omega\equiv\lied_k\omega\,,
\ee
are `new' KY and CCKY forms, respectively.
\end{prop}

\subsection{Graded symmetry operators}\label{ssc:grKYbr}
Particulary important symmetry operators are \emph{graded symmetry operators} (GSO) which graded (anti)-commute with the Dirac operator ${D}$. The {\em graded (anti)-commutator} of two homogeneous forms $\kappa$ and $\lambda$ of degrees ${p}$ and ${q}$ is defined as
\be\label{GC}
\bgc{\kappa}{\lambda}\equiv\kappa \lambda -(-1)^{pq}\lambda\kappa\,,\quad \bgc{\kappa}{\lambda}_{+}\equiv\kappa \lambda +(-1)^{pq}\lambda\kappa\,.
\ee
For inhomogeneous forms, the operation is generalized by linearity in both arguments.

\begin{prop}[Graded (anti)-commutation with ${D}$]\label{5.2}
The most general first-order operator $S$ which graded commutes with the Dirac operator $D$, i.e. obeys
\be
\bgc{D}{S}=DS-(\eta S)D=0\,,
\ee
is given in terms of an inhomogeneous KY form $\omega$ as follows:
\be\label{KKK}
S=K_\omega\equiv X^a\hook\omega\nabla_a+\frac{\pi-1}{2\pi}d\omega\,.
\ee
Similarly, the most general first-order operator $S$ which graded anticommutes with the Dirac operator $D$, i.e. obeys
\be
\bgc{D}{S}_{+}=DK_\omega+(\eta S)D=0\,,
\ee
is given in terms of an inhomogeneous CCKY form $\omega$ as:
\be\label{MMM}
S=M_\omega\equiv e^a\wedge\omega\nabla_a-\frac{\dm-\pi-1}{2(\dm-\pi)}\delta\omega\,.
\ee
\end{prop}
\begin{proof}
Both statements can be obtained as a combination of Prop.~\ref{2.2} and Prop.~\ref{anticomwD}.
\end{proof}

In the heuristic correspondence with the classical supersymmetric spinning particle, commuting GSO operators are related to observables Poisson-commuting with the superinvariant \cite{GibbonsEtal:1993}. Since their mutual Poisson brackets correspond to graded commutators we are motivated to investigate the graded commutation of commuting GSOs. As in the previous sections we restrict to the case where the graded commutator is of the first order. Since the product of two commuting GSOs is again a commuting GSO, the first order graded commutator has the form given by eq.~\eqref{KKK}. We can thus define a new operation on KY forms:
\begin{prop}[Graded KY bracket]\label{pr:grKYbr}
Let ${\kappa}$, ${\lambda}$ be two KY forms. Sufficient and necessary conditions that the graded commutator ${\bgc{K_\kappa}{K_\lambda}}$ remains of the first order are
\begin{equation}\label{algcondgrKYbr}
    \bgc{X^{(a}\hook\kappa}{X^{b)}\hook\lambda}=0\;,
\end{equation}
If these conditions are satisfied, we can define {graded Killing--Yano brackets ${\bgc{\kappa}{\lambda}_\KY}$} by relation
\begin{equation}\label{grKYbr}
    \bgc{K_\kappa}{K_\lambda}=K_{\gc{\kappa}{\lambda}_\KY}\;.
\end{equation}
The conditions \eqref{algcondgrKYbr} are equivalent to
\begin{equation}\label{algcondgrKYbrexpl}
    \sum_{k=0}\frac{(-1)^k}{(2k+1)!}\,
    \bigr( X^{(a}\hook\eta\kappa\bigl)\cwedge{2k{+}1}\bigl(X^{b)}\hook\lambda\bigr)=0\;.
\end{equation}
For homogeneous forms each term in the sum must vanish separately:
\begin{equation}\label{algcondgrhomKYbr}
    \bigr( X^{(a}\hook\kappa\bigl)\cwedge{2k{+}1}\bigl(X^{b)}\hook\lambda\bigr)=0\,,\quad
    k=0,1,\dots\,.
\end{equation}
\end{prop}
\begin{proof}
The derivation of these statements is analogous to those in Secs.~\ref{ssc:commuting} and \ref{ssc:foc}.
\end{proof}
\noindent Note that for homogeneous forms as a consequence of the conditions \eqref{algcondgrKYbrexpl} we also get relations
\begin{equation}\label{ctrcondgrhomKYbr}
    \kappa\cwedge{2k}\lambda = 0\;,\qquad k=1,2,\dots \;.
\end{equation}

Similarly to Sec.~\ref{ssc:KYbr}, we can derive the explicit expressions for graded Killing--Yano brackets
\begin{equation}
    \bgc{\kappa}{\lambda}_\KY =
     \frac1\pi\,\delta\,\sum_{k=0}\frac{(-1)^k}{(2k{+}1)!}\,
     \bigl(\pi\eta\kappa\bigr)\cwedge{2k}\bigl(\pi\lambda\bigr)\;,
\end{equation}
which for homogeneous forms of degrees ${p}$ and ${q}$ reduces to
\begin{equation}\label{grKYbrPot}
    \bgc{\kappa}{\lambda}_\KY =
     \frac{pq}{p+q-1}\,\delta\,\Bigl((\eta\kappa)\wedge\lambda\Bigr)\;.
\end{equation}
It is straightforward to check that ${\bgc{\lambda}{\kappa}_\KY=(-1)^{pq+p+q}\bgc{\kappa}{\lambda}_\KY}$. This is consistent with the property of the graded commutator of corresponding operators ${K_\kappa}$, ${K_\lambda}$, which have opposite parity to that of ${\kappa}$ and ${\lambda}$.

\subsection{Relation to Schouten--Nijenhuis brackets}\label{ssc:SNbr}

It is well known that skew-symmetric tensors form a (graded) Lie algebra with respect to the skew-symmetric Schouten--Nijenhuis (SN) bracket \cite{Schouten:1940, Schouten:1954, Nijenhuis:1955}.
This fact led Kastor, Ray, and Traschen \cite{KastorEtal:2007} to investigate whether, similar to Killing vectors, KY tensors form a subalgebra of this algebra. Unfortunately, such statement is not true in general, the authors were able to give two counter examples disproving the conjecture. On the other hand, the statement is true in maximally symmetric spaces. In particular, in Minkowski space KY tensors form a Lorentz-like subalgebra with respect to the SN brackets \cite{KastorEtal:2007}. It is then natural to ask under which conditions
SN brackets have to produce KY tensors.

In this subsection we shall demonstrate that graded Killing--Yano brackets, defined in Prop.~\ref{pr:grKYbr}, reduce for homogeneous forms to corresponding SN brackets. This means that restriction \eqref{algcondgrhomKYbr} represents a sufficient condition under which SN brackets of two KY forms give rise to a KY form.

Let us first recall the definition of SN brackets and state some useful identities.
SN brackets are usually defined for (skew-symmetric) multivector fields.
Given $A$ a rank $p$ multivector and $B$ a rank $q$ multivector, their Schouten--Nijenhus bracket is a $(p+q-1)$-multivector defined by
\be
[A,B]_{\SN}^{a_1\dots a_{p+q-1}}=pA^{b[a_1\dots a_{p-1}}\nabla_b B^{a_p\dots a_{p+q-1}]}
+ (-1)^{pq} qB^{b[a_1\dots a_{q-1}}\nabla_b A^{a_q\dots a_{p+q-1}]}\, . \label{SN_definition}
\ee
In the case of a vector $A=k$ this definition reduces to the Lie-derivative of a covariant antisymmetric tensor, $[k,B]_{\SN}=\lied_k B$.

For our purpose it is useful to extend the definition of the SN bracket to forms. For a $p$-form $\alpha$ and a $q$-form $\beta$ we define
\be
[\alpha,\beta]_{\SN}\equiv\Bigl([\alpha^\sharp, \beta^\sharp]_{\SN}\Bigr)^\flat\,.
\ee
In our notation this is
\be\label{NS2}
\frac{(p+q-1)!}{p!\,q!}\,[\alpha,\beta]_{\SN}=(X^a\hook\alpha)\wedge \nabla_a\beta+(-1)^{pq}(X^a\hook \beta)\wedge \nabla_a\alpha\,.
\ee
In particular, for a conformal Killing 1-form $\kappa$ and an arbitrary $q$-form $\beta$ we have
\be\label{SN_ckv}
[\kappa, \beta]_{\SN}
%=\Bigl([\kappa^\sharp,\beta^\sharp]_{\SN}\Bigr)^\flat
=(\lied_{\kappa^\sharp}\beta^\sharp)^\flat=\lied_{\kappa^\sharp} \beta+\frac{2q}{n}\beta\, \delta k \,,
\ee
where in the last equality we have used the fact that $\lied_{\kappa^\sharp}g=-\frac{2}{n}(\delta\kappa) g$. When $\kappa$ is a Killing 1-form, $\delta \kappa=0$, and the previous formula reduces to
\be\label{SN_kv}
[\kappa, \beta]_{\SN}=\lied_{\kappa^\sharp} \beta\,.
\ee

It may be also useful to rewrite formula \eqref{NS2} as
\be
\frac{(p+q-1)!}{p!\,q!}\,[\alpha,\beta]_{\SN}=-\delta(\alpha\wedge \beta)+\delta\alpha\wedge \beta+(-1)^p\alpha\wedge\delta\beta\,,\label{eq:useful_NS}
\ee
which follows from identity \eqref{A23}.
For KY forms $\kappa$ and $\lambda$ this reduces to
\be
[\kappa,\lambda]_{\SN}=-\frac{p!\,q!}{(p+q-1)!}\,\delta(\kappa\wedge \lambda)\,. \label{eq:useful_KY}
\ee
Comparing this identity with \eqref{grKYbrPot} allows us to formulate a direct relation between SN and graded Killing--Yano brackets
\begin{prop}[SN and graded Killing--Yano brackets]\label{pr:grKYSN}
For homogeneous KY ${p}$-form ${\kappa}$ and \mbox{${q}$-form} ${\lambda}$ satisfying conditions \eqref{algcondgrhomKYbr}, the graded Killing--Yano brackets are related to SN brackets as
\begin{equation}\label{grKYSN}
   \bgc{\kappa}{\lambda}_\KY =
     -\frac{(p+q-2)!}{(p{-}1)!\,(q{-}1)!}\,\bigl[\eta\kappa,\lambda\bigr]_\SN\;.
\end{equation}
\end{prop}

As a corollary, we see that the algebraical conditions \eqref{algcondgrhomKYbr} guarantee that the SN brackets of homogeneous KY forms produce again a KY form. Let us take an example directly related to investigations in \cite{KastorEtal:2007}. For two KY 2-forms the conditions \eqref{algcondgrhomKYbr} claim that the associated Killing tensor vanishes, $K_{ab}\equiv\kappa_{(a|c|}\lambda_{b)}{}^c=0$ (where, e.g., ${\kappa_{ac}=X_c\hook X_a\hook\kappa}$). In such a case $[\kappa,\lambda]_{\SN}$ must be a KY 3-form.
Condition $K_{ab}=0$ is, for example, satisfied in four dimensional Minkowski space for $\kappa=dt\wedge dx$ and $\lambda=dy\wedge dz$. In this case, however, we have $[\kappa,\lambda]_{\SN}=0$. We do not know any nontrivial example for which $K_{ab}=0$ and $[\kappa,\lambda]_{\SN}$ is nontrivial.

\subsection{Hodge duality}\label{ssc:HodgeDual}

SN brackets can be related also to Killing--Yano brackets of CCKY forms discussed in Sec.~\ref{sc:comopKYbr}. It can be done employing the well-known duality between KY and CCKY forms. We begin by reviewing some facts about Hodge duality.

The Hodge dual of a homogeneous ${p}$-form can be written as
${*\omega = \frac1{p!}\,{\displaystyle\omega\cwedge{p}\eps}}$,
%\begin{equation}\label{Hodgedual}
%    *\kappa = \frac1{p!}\kappa\cwedge{p}\eps\;.
%\end{equation}
where, ${\eps}$ is Levi-Civita tensor---${\dm}$-form normalized as ${{\displaystyle\eps\cwedge{\dm}\eps}=s\dm!}$, ${s}$ being a product of signs in the signature. It can be rewritten using Clifford multiplication as
\begin{equation}\label{HodgeDualCl}
    *\omega=(-1)^{(n{-1})p+[\frac p2]}\,\eps\omega = (-1)^{[\frac p2]}\,\omega\eps\;,
\end{equation}
which can be generalized directly to inhomogeneous forms by linearity. It means that the Hodge dual differs from the multiplication by ${\eps}$ just by a sign. Therefore we will understand by Hodge duality any of these two operations.

It is well known that the space of CKY forms is invariant under Hodge duality and in particular that KY forms are Hodge dual to CCKY forms. Namely, let $\kappa$ be a KY form and $\alpha$ its Hodge dual CCKY form. This relation translates to an analogous relation for the corresponding operators ${K_\kappa}$ and ${M_\alpha}$, namely
\begin{equation}\label{KMduality}
    K_{\eps\alpha} = (-1)^{\dm-1}\eps\,M_\alpha\;,\qquad
    M_{\eps\kappa} = (-1)^{\dm-1}\eps\,K_\kappa\;.
\end{equation}
It is a direct consequence of identities \eqref{HDwedgehook} and \eqref{HDddelta}.

Now we can formulate properties of the Killing--Yano brackets under Hodge duality.
\begin{prop}[Hodge duality of Killing--Yano brackets]\label{pr:HDKYbr}
Let ${\mu}$ be an odd KY form, and ${\alpha}$, ${\beta}$ even CCKY forms dual to KY forms ${\kappa}$ and ${\lambda}$ as ${\alpha = \eps\kappa}$, ${\beta = \eps\lambda}$. Then the Killing--Yano brackets defined in Sec.~\ref{sc:comopKYbr} are related to Killing--Yano brackets from Sec.~\ref{ssc:genKYbr} as
\begin{equation}\label{KYbrduality}
    [\mu,\alpha]_\KY = (-1)^{n{-}1}\eps\, [\mu,\kappa]_\KY\;,\qquad
    [\alpha,\beta]_\KY = (-1)^{n{-}1}\eps^2\, [\kappa,\lambda]_\KY\;.
\end{equation}
Note that ${\eps^2}$ contributes only a sign factor: ${\eps^2=(-1)^{[\frac\dm2]}s}$, cf.~\eqref{epsids}.
\end{prop}
\begin{proof}
The brackets on left hand side are defined in terms of commutators of the operators ${K_\mu}$, ${M_\alpha}$, and ${M_\beta}$. For example, for the first relation we have
\begin{equation}\label{HDKYbrproof1}
    M_{[\mu,\alpha]_\KY}=[K_\mu,M_{\eps\kappa}]=(-1)^{n-1}\bigl(K_\mu\eps K_\kappa - \eps K_\kappa K_\mu\bigr)\;,
\end{equation}
where we used \eqref{KYbrackets0} and \eqref{KMduality}. Taking into account the property \eqref{epscom} we obtain
\begin{equation}\label{HDKYbrproof2}
    M_{[\mu,\alpha]_\KY}=(-1)^{n-1}\eps\, [K_\mu,K_\kappa]=(-1)^{n-1}\eps\, K_{[\mu,\kappa]_\KY}\;,
\end{equation}
which concludes the proof. The second relation can be demonstrated analogously.
\end{proof}

Finally, using this duality we can express some of the Killing--Yano brackets of CCKY forms in terms of the SN brackets.
\begin{prop}[KY and SN brackets]\label{Prop3.4}
Under the necessary conditions \eqref{necesary_cond} and for homogeneous odd KY ${\kappa}$, ${\lambda}$, ${\mu}$ of degrees ${p}$, ${q}$, and ${r}$, respectively, and even CCKY forms ${\alpha}$, ${\beta}$, ${\omega}$, of degrees $a$, ${b}$, and ${c}$, respectively, the Killing--Yano brackets \eqref{KYbrackets0} are related to the SN brackets as follows
\begin{subequations}\label{KY-NS_Hom}
\begin{align}
   [\kappa,\lambda]_\KY &=
     \frac{(p+q-2)!}{(p{-}1)!\,(q{-}1)!}\,[\kappa,\lambda]_\SN\;,
     \label{KY-NS_HomKK}\\
   [\mu,\omega]_\KY &=
     s\,(-1)^{[\frac{n{-}1}2]}\,\frac{(r+n-c-2)!}{(r{-}1)!\,(n{-}c{-}1)!}\,\eps\,[\mu,\eps\omega]_\SN\;,
     \label{KY-NS_HomKM}\\
   [\alpha,\beta]_\KY &=
     s\,(-1)^{[\frac{n{-}1}2]}\,\frac{(2n-a-b-2)!}{(n{-}a{-}1)!\,(n{-}b{-}1)!}\,[\eps\alpha,\eps\beta]_\SN\
     \qquad\text{for ${\dm}$ odd}\;.
     \label{KY-NS_HomMM}
   \end{align}
\end{subequations}
\end{prop}
\begin{proof}
The statement of the proposition is a combination of the duality formulated in Prop.~\ref{pr:HDKYbr} and relation between graded Killing--Yano brackets and SN brackets from Prop.~\ref{pr:grKYSN}. The only restriction arises from the fact that we have to identify the KY bracket of KY forms with their graded KY bracket. It can be done if at least one of the KY form is odd (as can be seen from relations between commutator and graded commutator of the corresponding operators). The problematic case arise in an even dimension, when the KY bracket of two even CCKY forms ${\alpha}$ and ${\beta}$ transfers to the KY bracket of two even KY forms ${\eps\alpha}$ and ${\eps\beta}$. Such a bracket is equivalent to graded anticommutator and cannot be thus rewritten using SN brackets.
\end{proof}

\subsection{Conformal symmetries}\label{ssc:confsym}
Until now we have discussed only Dirac symmetry operators which (anti)-commuted or graded (anti)-commuted with ${D}$. This discussion can be naturally extended
to general symmetry operators R-commuting with ${D}$, as described by Prop.~\ref{1.1}. We restrict here only to a useful example generalizing the notion of Lie derivative along a conformal Killing vector.
\begin{prop}\label{pr:CKV}
Let ${\kappa}$ be a CKY 1-form, ${\omega}$ an arbitrary CKY form, and $S_\kappa$ and $S_\omega$ be the corresponding symmetry operators of the Dirac operator given by Prop.~\ref{1.1}. Then the commutator of these operators remains of the first-order
and reads
\be\label{kf}
  [S_{\kappa}, S_\omega]=S_{\dot \omega}\,,\quad \dot \omega=\lied_{\kappa^\sharp}\omega+\frac{\pi+1}{n}\delta \kappa\, \omega\,.
\ee
As a corollary we have that for any $\kappa$ CKY 1-form and ${\omega}$ CKY form
\be\label{cor55}
\dot \omega=[\kappa,\omega]_{\SN}-\frac{\pi-1}{n}\delta\kappa\, \omega=\lied_{\kappa^\sharp}\omega+\frac{\pi+1}{n}\delta \kappa\, \omega\,
\ee
is a new CKY form.
\end{prop}
\begin{proof}
We calculate (denoting by $\kappa^a=X^a\hook \kappa$ and $\omega^a=X^a\hook \omega$)
\be
[S_\kappa, S_\omega]=\Bigl\{\kappa^a(\nabla_a \omega^b)-\omega^a(\nabla_a \kappa^b)+\frac{1}{4}[d\kappa,\omega^b]\Bigr\}\nabla_b+
\mbox{zeroth order terms}.
\ee
Since this commutator is a symmetry operator of the Dirac operator and it is of the first-order, it must be of the form described by Prop.~\ref{1.1}. So we must have $[S_\kappa, S_\omega]=S_{\dot \omega}+\alpha D$, where $S_{\dot \omega}$ is given by \eqref{SOProp1} and $\alpha$ is some form. This leads to an equation
\be
\kappa^a(\nabla_a \omega^b)-\omega^a(\nabla_a \kappa^b)+\frac{1}{4}[d\kappa,\omega^b]=X^b\hook (\dot \omega-\eta \alpha)+e^b\wedge(\eta\alpha)\,.
\ee
By taking $X_b \hook$ of this equation we find that $\alpha=0$. So we have
\ba\label{nic}
\dot \omega&=&\frac{1}{\pi}e^b\wedge \Bigl\{\kappa^a(\nabla_a \omega^b)-\omega^a(\nabla_a \kappa^b)+\frac{1}{4}[d\kappa,\omega^b]\Bigr\}=
\kappa^a\nabla_a\omega-\nabla_a\kappa\wedge\omega^a-\frac{\pi-1}{n}\delta \kappa\omega\nonumber\\
&=&[\kappa,\omega]_{\SN}-\frac{\pi-1}{n}\delta\kappa \omega=
\lied_{\kappa^\sharp}\omega+\frac{\pi+1}{n}\delta \kappa \omega\,.
\ea
In the third equality we have used the definition of SN brackets, \eqref{NS2}, and in the last we have applied Eq. \eqref{SN_ckv}.
\end{proof}
\noindent The corollary of this proposition, Eq. \eqref{cor55}, is already known in the literature, see, e.g., \cite{BennCharlton:1997}. It is also a generalization of relation \eqref{new2} to the case of a conformal Killing vector.

%%%%%%%%%%%%%%%%%%%%%%%%%%%%%%%%%%%%%%%%%%%%%%%%%%%%%%%%%%%%%%%%%%%%%%%%%%
%%%%%%%%%%%%%%%%%%%%%%%%%%%%%%%%%%%%%%%%%%%%%%%%%%%%%%%%%%%%%%%%%%%%%%%%%%

\section{Discussion and Conclusions}\label{sc:concl}
In this paper we have studied first-order symmetry operators of the Dirac operator, in all dimensions and signatures, both from the point of view of their general properties and in the specific context of Kerr-NUT-AdS black hole spacetimes admitting the PCKY tensor.
Namely, we have established
the general form of first-order operators that commute, anti-commute, graded commute and graded anti-commute with the Dirac operator. Our main attention focused on
the `commuting operators'. Such operators are directly related to separability of the Dirac equation. We have demonstrated that commuting operators assume a unique form given by Prop.~\ref{2.2};
they split into Clifford even and Clifford odd parts and correspond to odd KY and even CCKY inhomogeneous forms, respectively.
We have further studied the commutators of these operators and given algebraic conditions under which they remain of the first-order. In that case the explicit expression for the commutator is given by Eq. \eqref{KYbrackets0}.
In particular, when this expression vanishes the operators, besides commuting with the Dirac operator, commute also between themselves.

As a main application of the general theory we have demonstrated that in the most general known (spherical) Kerr-NUT-(A)dS black hole spacetimes in all dimensions
there exists a complete set of first-order mutually commuting operators, one of which is the Dirac operator.
These operators correspond to the tower of symmetries generated from the PCKY tensor. In even dimensions they are given by the set \eqref{symop} whereas in  odd dimensions one can choose any complete subset of operators \eqref{symop} and \eqref{symop2}.
Such operators underlie separability of the massive Dirac equation in these spacetimes demonstrated by Oota and Yasui in all dimensions \cite{OotaYasui:2008}. This will be further discussed in a forthcoming paper \cite{CKK:separability}.

As a by-product of our construction we are able to partly address the issue of whether or not CKY tensors `behave as proper symmetries'.
(Conformal) KY tensors are known to naturally generalize (conformal) Killing symmetries to higher-rank antisymmetric objects; they obey a natural generalization of the conformal Killing vector equation  \cite{Kashiwada:1968}, similar to Killing vectors they give rise to various conserved quantities---Y-ADM charges for example \cite{KastorTraschen:2004}. For this reason such forms are expected to be associated with symmetries in some `appropriately generalized sense'. In particular one might expect that they should form a closed Lie algebra with respect to some appropriately chosen {\em 'Killing--Yano bracket'}.
This is true for Killing vectors which close on themselves with respect to the Lie bracket as well as for (symmetric) Killing tensors which are known to form such an algebra with respect to the symmetric Schouten--Nijenhuis bracket. It would seem that an obvious candidate for the Killing--Yano bracket is the antisymmetric Schouten--Nijenhuis bracket. Unfortunately, it was shown in \cite{KastorEtal:2007} that, except in maximally symmetric spacetimes, this is not necessarily the case.

Our construction enables us to discuss the Killing--Yano bracket in a specific setting.
 Namely, we were able to introduce a bilinear operation \eq{KYbrackets0} acting on the space of odd KY and even CCKY forms, provided that the conditions \eq{fordcond} are satisfied.
By considering more general first-order symmetry operators of the Dirac operator, this definition can be further extended to `arbitrary' CKY forms.
Unfortunately, the derived algebraic conditions are rather stringent and at present we do not know any examples in which these conditions are satisfied and at the same time the Killing--Yano brackets are non-vanishing. This raises the following two interesting possibilities:
1) Algebraic conditions allow for non-vanishing Killing--Yano brackets in suitable spacetimes. If this is the case, the introduced Killing--Yano bracket is a non-trivial operation and can be used to study the symmetries of the Dirac operator: to this effect we have presented complete formulas for the calculation of the brackets.
2) Algebraic conditions are so strong that, when satisfied, they automatically imply zero Killing--Yano brackets and therefore a set of (graded) commuting operators. In this case the algebraic conditions turn into a useful tool to find constants of motion.
This possibility is further supported by a number of examples in flat and pp-wave spacetimes as well as by a highly non-trivial
example of symmetries of Kerr-NUT-AdS spacetimes in all dimensions.
Let us stress that, however strong, the presented algebraic conditions are {\em necessary conditions} which have to be satisfied in order for the (graded) commutators of the studied first-order symmetry operators of the Dirac operator to vanish.

Let us make two more remarks on the importance of the antisymmetric  Schouten--Nijenhuis brackets.
1) The antisymmetric Schouten--Nijenhuis brackets play an important (though not entire) role for the commutation of first-order symmetry operators
commuting with the Dirac operator, see Prop~\ref{Prop3.4}. Hence such brackets are important for the theory of separability of the Dirac equation.
This is to be compared with the separability theory for the scalar field in which it is the symmetric Schouten-Nijenhuis brackets which play the key role.
2) Vanishing of the Schouten--Nijenhuis bracket of two Killing--Yano tensors is a necessary requirement for the graded commutation of the corresponding commuting GSOs
\eq{KKK}, see Prop.~\ref{pr:grKYSN}. This has an important consequence for the theory of spinning particles \cite{GibbonsEtal:1993}.
Namely, at the spinning particle level the commuting GSOs correspond to superinvariants, i.e., to quantities (Dirac-Poisson)-commuting with the generic supercharge.
Such quantities are determined by Killing--Yano tensors, in a way similar to Eq. \eq{KKK}, see \cite{GibbonsEtal:1993}. Hence,  vanishing of the Schouten--Nijenhuis bracket
of two Killing--Yano tensors is a necessary requirement for two  non-generic supercharges to (Dirac-Poisson)-commute.
Since the (Dirac-Poisson)-commutation at the spinning particle level corresponds to {\em graded commutation}, not commutation, the
existence of the complete set of commuting operators does not necessarily imply the existence of a corresponding set of the (Dirac-Poisson)-commuting
spinning particle observables. This explains the `discrepancy' between our operator results presented in Sec.~\ref{apx:setcomop} and the recent spinning particle analysis of Ahmedov and Aliev \cite{AhmedovAliev:2009}.

Let us finally mention that in our study we have concentrated on the `standard' Dirac operator and its symmetries.
In spacetimes relevant for string theory and various supergravities, however, the Dirac operator gets modified by the presence of fluxes.
The first-order symmetry operators
of the flux-modified Dirac operator have been studied
recently in \cite{AcikEtal:2008c, HouriEtal:2010a, KubiznakEtal:2010}; they correspond to symmetries which generalize CKY tensors
in the presence of fluxes
\cite{Wu:2009a, KubiznakEtal:2009b, KubiznakEtal:2010}.
Of special importance seems the torsion-modified Dirac operator. Separability of the corresponding Dirac equation has been demonstrated
in the most general spherical black hole spacetime of minimal gauged supergravity \cite{Wu:2009a}, or in the background of the Kerr--Sen geometry and its higher-dimensional generalizations \cite{HouriEtal:2010b}. We expect that the results of the present paper can be straightforwardly generalized
to these more general setups.

\section*{Acknowledgments}
We would like to thank G.~W.~Gibbons and C.~M.~Warnick for useful discussions.
M.C. would like to thank DAMTP at the University of Cambridge for the hospitality during part of the work, and is grateful to the Federal University of Ouro Preto for financial support for the travel.
P.K. was supported by Grant No. GA\v{C}R-202/08/0187,
Project No.~MSM0021620860 and appreciates the hospitality of the Theoretical
Physics Institute of the University of Alberta where a part of the work has been done.
D.K. is grateful to Herchel Smith Postdoctoral Fellowhip at the University of Cambridge for financial support.

%%%%%%%%%%%%%%%%%%%%%%%%%%%%%%%%%%%%%%%%%%%%%%%%%%%%%%%%%%%%%%%%%%%%%%%%%%
%%%%%%%%%%%%%%%%%%%%%%%%%%%%%%%%%%%%%%%%%%%%%%%%%%%%%%%%%%%%%%%%%%%%%%%%%%
%%%%%%%%%%%%%%%%%%%%%%%%%%%%%%%%%%%%%%%%%%%%%%%%%%%%%%%%%%%%%%%%%%%%%%%%%%
%%%%%%%%%%%%%%%%%%%%%%%%%%%%%%%%%%%%%%%%%%%%%%%%%%%%%%%%%%%%%%%%%%%%%%%%%%

\section*{Appendices}
\appendix

%%%%%%%%%%%%%%%%%%%%%%%%%%%%%%%%%%%%%%%%%%%%%%%%%%%%%%%%%%%%%%%%%%%%%%%%%%
%%%%%%%%%%%%%%%%%%%%%%%%%%%%%%%%%%%%%%%%%%%%%%%%%%%%%%%%%%%%%%%%%%%%%%%%%%

\section{Conventions and useful identities}\label{apx:notation}

\subsection{Notes on the formalism}

In the paper we follow the notation and formalism of \cite{BennTucker:book, HouriEtal:2010a}, to which we refer the reader.\footnote{Note, however, that we use a convention different from \cite{BennTucker:book} for the wedge product, the inner derivative and related operations. Namely we use the standard convention
${(\alpha\wedge\beta)_{a\dots b\dots} = \frac{(p+q)!}{p!\,q!}\,\alpha_{[a\dots}\,\beta_{b\dots]}}$.
Nevertheless, both conventions are isomorphic and, therefore, the most of formulas of \cite{BennTucker:book} is valid in our convention without any changes.}
Here, we are going to make some comments on geometrical meaning of this formalism.

Spacetime is a (pseudo-)Riemannian spin manifold $M$ of dimension ${\dm}$ with metric ${g_{ab}}$. We use Latin indices to denote spacetime tensors. We assume that on this manifold we can build the Dirac bundle of spinors and Clifford bundle with the irreducible representation on the Dirac bundle. Each fiber of the Clifford bundle has structure of the Clifford algebra generated by the \emph{gamma matrices ${\gamma^a}$}, which also connect the Clifford bundle with the tangent space. The gamma matrices satisfy the standard relation
\begin{equation}\label{ggmetric}
    \gamma^a\,\gamma^b + \gamma^b\, \gamma^a = 2g^{ab}\;,
\end{equation}
which allows to eliminate any symmetric product of ${\gamma}$'s and reduce any Clifford object ${\slash\mspace{-10mu}\omega}$ to a sum of antisymmetric products ${\gamma^{a_1\dots a_p}=\gamma^{[a_1}}\dots\gamma^{a_p]}$ with proper coefficients given by tangent antisymmetric forms ${\omega^{(p)}_{a_1\dots a_p}}$,
\begin{equation}\label{clobrepr}
    \slash\mspace{-10mu}\omega
     = \sum_p \frac1{p!}\, \omega^{(p)}_{a_1\dots a_p} \gamma^{a_1\dots a_p}\;.
\end{equation}
This representation is unique and gives thus an isomorphism ${\gamma_*}$ of the Clifford bundle with the exterior algebra ${\Omega(M)=\bigoplus_{p=0}^{\dm} \Omega^p(M)}$ of inhomogeneous antisymmetric forms,
\begin{equation}\label{cleaiso}
    \slash\mspace{-10mu}\omega = \gamma_* \omega \;,\qquad
    \text{where}\quad \omega = \sum_p  \omega^{(p)}\;.
\end{equation}

Since we do not need details of the action of the Clifford bundle on the Dirac spinors and we work mainly with objects from the Clifford bundle, we can use the isomorphism \eqref{cleaiso} and replace the Clifford bundle by the exterior algebra (cf., e.g., \cite{BennTucker:book}). However, when we use the antisymmetric forms in sense of objects from the Clifford bundle, we denote them without spacetime indices.\footnote{Since we work with inhomogeneous forms, the indices would not be very helpful. However, for a homogeneous form ${\omega}$, by ${\omega_{a_1\dots a_p}}$ we mean its tensor components.} Analogously, we write vectors without indices if we understand them as elements of dual to ${\Omega^1(M)}$. The metric allows to raise and lower indices: if $\omega$ is a 1--form and $v$ a vector, we denote the corresponding vector and 1--form as ${\omega^\sharp}$ and ${v^\flat}$, respectively. The notion of these operations can be naturally extended to higher rank tensors.

As we said, the gamma matrices map a covector ${\alpha_a}$ to a Clifford object ${\slash\mspace{-10mu}\alpha=\alpha_a\gamma^a=\gamma_*\alpha}$. When we restrict to the exterior algebra picture, the isomorphism \eqref{cleaiso} gives a trivial map ${\alpha_a \to \alpha = \alpha_a e^a}$. This trivial identification of the tangent covectors with 1-forms is provided by the \emph{canonical 1-form} ${e^a\in TM\otimes\Omega^1(M)}$ which plays the role of the gamma matrices.\footnote{The canonical form ${e^a}$ is actually the identity tensor with one explicit vector index and a hidden form index. If we write the form index explicitly, we get ${e^a_b = \delta^a_b}$.} We introduce also a dual object ${X_a}$ which maps ${\Omega^1(M)\to T^*M}$ as ${\alpha\to \alpha_a=X_a\hook\alpha}$.

We have introduced the canonical form ${e^a}$ as a tensor in the spacetime index ${a}$. However, it is customary to chose an orthonormal vielbein of 1-forms ${\efr^a}$ and a dual vector frame~${\Xfr_a}$, and to express all spacetime indices (but not hidden form indices) with respect to these frames. In such a case, the components of the canonical form are just the vielbein forms, ${e^a = \efr^a}$, and similarly ${X_a = \Xfr_a}$. Nevertheless, one has to remember that in the case of the canonical form ${e^a}$ and of its dual ${X_a}$,  the index ${a}$ is a tensor component. It will be significant when taking the covariant derivative; see discussion below.

The `hook' operation (\emph{inner derivative}) is an action of a vector ${v}$ on any antisymmetric form ${\omega}$. When transformed back to tangent tensors, it is just a contraction of the vector with the form in the first index,
\begin{equation}\label{hook}
    (v\hook\omega)_{a_1\dots a_{p{-}1}}=v^b\omega_{b a_1\dots a_{p{-}1}}\;.
\end{equation}
For a scalar ${\varphi}$, we set ${v\hook\varphi=0}$. We also use freely ${e_a=g_{ab}e^b}$, ${X^a=g^{ab}X_b}$, ${X_a\hook e^b = \delta_a^b}$, and ${e^\sharp_a = X_a}$.

The Clifford bundle is endowed with the Clifford product. Using the isomorphism \eqref{cleaiso} we can define it also in the exterior algebra ${\Omega(M)}$ and we denote it by juxtaposition of forms. Using the relation \eqref{ggmetric}, or its equivalent ${e^a\,e^b + e^b\, e^a = 2g^{ab}}$, we can express the Clifford product in terms of the exterior product ${\wedge}$ and the inner derivative ${\hook}$. In particular, if $\alpha$ is a $1$-form and $\omega$ a $p$-form, we get \cite{BennTucker:book, HouriEtal:2010a}
\begin{equation}\label{eq:notation_product}
\begin{aligned}
\alpha\omega &= \alpha \wedge \omega + \alpha^\sharp \hook \omega \;,\\
\omega\alpha &= (-1)^p \bigl( \alpha \wedge \omega - \alpha^\sharp \hook \omega\bigr)\;.
\end{aligned}
\end{equation}
Applying recursively eq.~\eqref{eq:notation_product}, it is possible to construct the product of two generic forms \eqref{usefulProduct1} and \eqref{usefulProduct2} and further useful formulas listed in Appendix~\ref{apx:usefulid}. In these identities, it is useful to introduce the contracted wedge product defined inductively by
\ba
\alpha\cwedge{0}\beta &=& \alpha \wedge \beta  \, , \nn \\
\alpha\cwedge{k}\beta &=& (X_a \hook \alpha ) \cwedge{k-1} (X^a \hook \beta) \qquad (k\ge 1) \, , \nn \\
\alpha\cwedge{k}\beta &=& 0 \qquad\qquad\qquad\qquad\qquad\, (k<0) \, .\label{cwedgedef}
\ea
The contracted wedge product vanishes, ${\displaystyle\alpha\cwedge{k}\beta=0}$, if $k<p+q-\dm$ or ${k>\min(p,q)}$, where ${\alpha\in\Omega^p(M)}$ and ${\beta\in\Omega^q(M)}$.

An inhomogeneous form $\omega$ can be written as a sum of homogeneous $p$-forms
\begin{equation}\label{inhomforms}
\omega = \sum_{p=0} \omega^{(p)} \; ,
\end{equation}
where throughout the paper we use the convention that the upper limit of sums of forms, when not indicated, is the highest value of the summation index that yields a non-zero result in the argument of the sum. The degree operator ${\pi}$ and parity operator ${\eta}$ acts on an inhomogeneous form $\omega$ as
\begin{equation}\label{pietadef}
\pi \omega = \sum_{p=0} p \, \omega^{(p)} \;,\quad
\eta \omega = \sum_{p=0} (-1)^p \omega^{(p)} \, .
\end{equation}

For ${A}$, ${B}$ Clifford valued operators we will indicate with $[A,B]\equiv[A,B]_-$ and ${[A,B]_+}$ the commutator and, respectively, the anticommutator
\be
[A,B] = AB - BA \; , \quad [A,B]_+ = AB + BA \; .
\ee
For operators ${A}$, ${B}$ with defined parity given by ${\pi_{A,B}=0,1}$ as ${\eta A = (-1)^{\pi_A} A}$,  ${\eta B} = {(-1)^{\pi_B} B}$, the graded (anti)commutator is
\begin{equation}\label{gr(anti)com}
 \bgc{A}{B} = AB- (-1)^{\pi_A\pi_B} BA\;,\quad
\bgc{A}{B}_+= AB+(-1)^{\pi_A\pi_B} BA\;.
\end{equation}

The metric covariant derivative can be lifted into the Clifford bundle demanding ${\nabla_a\gamma^b=0}$. Translating this relation into the exterior algebra, we obtain natural relations
\begin{equation}\label{nablaeX=0}
    \nabla_a e^b =0\;,\quad \nabla_a X_b =0\;,
\end{equation}
which just say that we use the same derivative in the tangent space and in the exterior algebra.
In this formalism the Dirac operator is written as ${D\equiv e^a\nabla_{\!a}=\nabla_{\!a} e^a}$, the exterior derivative acting on forms as ${d = e^a \wedge \nabla_{\!a} = \nabla_{\!a}\, e^a \wedge}$, and the co-differential  ${\delta = - X^a \hook \nabla_{\!a} = -\nabla_{\!a}\, X^a \hook}$. All these expressions are to be understood as operators acting to the right. Using \eqref{eq:notation_product}, we get the well-known action of the Dirac operator on the exterior algebra
\begin{equation}\label{diracddelta}
    D\omega = \nabla_{\!a}(e^a\omega)
      = \nabla_{\!a}(e^a\wedge\omega) + \nabla_{\!a}(X^a\hook\omega)
      = d\omega - \delta\omega\;.
\end{equation}

The Latin indices used in the paper can be understood as abstract indices indicating tensor character of the spacetime tensors. Or, as we mentioned above, they can denote vielbein components with respect to an orthogonal frames ${\efr^a}$.
%In such a case, the canonical objects ${X_a}$ and ${e^a}$, ${a=1,2,\dots,\dm}$ can be identify with the vielbein vectors and 1-forms and metric vielbein coefficients are ${g_{ab}=\eta_{ab}}$.
In that case, we shall stress that by ${\nabla_{\!a} v^b}$ we mean vielbein components of the covariant derivative
\begin{equation}\label{covdervcomp}
    \nabla_{\!a}v^b = X_a[v^b] + \omega{}_a{}^b{}_c v^c\;,
\end{equation}
where ${X_a[v^b]}$ are derivatives of the components ${v^b}$ in directions ${X_a}$ and ${\omega{}_a{}^b{}_c}$ are connection (Ricci) coefficients. Thanks to the orthonormality of the vielbein, one has ${\omega_{abc} = -\omega_{acb}}$ and we can define connection 2-forms ${\omega_a= \frac12\omega_{abc} e^b\wedge e^c}$. For a generic form ${\alpha}$, the covariant derivative can then be rewritten as
\begin{equation}\label{covdervielbein}
    \nabla_{\!a}\alpha = X_a[\alpha] - \omega{}_a \cwedge{1} \alpha
       %= X_a[\alpha] + \bigl[-{\textstyle\frac12}\omega_a,\alpha\bigr]\;.
\end{equation}
Here, ${(X_a[\alpha])_{b_1\dots b_p} = X_a[\alpha_{b_1\dots b_p}]}$.
With this convention we recover
\begin{equation}\label{nablae=0invielbein}
    \nabla_{\!a} e^b = X_a[e^b] + \omega{}_a{}^b{}_c\, e^c - \omega_a\cwedge{1} e^b = 0\;,
\end{equation}
since ${X_a[e^b]=0}$ and ${\displaystyle\omega_a \cwedge{1}e^b = \omega{}_a{}^b{}_c\, e^c}$.
This should not be confused with the standard relation
\begin{equation}\label{nablae=omegae}
    \nabla_{\!a} \efr^b = -\omega{}_a{}^b{}_c\, \efr^c \;,
\end{equation}
in which one ignores the convention \eqref{covdervcomp}, since here the index ${b}$ is not a tensor component but the index labeling the vielbein 1-forms.

\subsection{Clifford and contracted wedge products}\label{apx:usefulid}
Here we gather useful identities related to the Clifford and contracted wedge products used in the main text.
Some of the relations are taken from \cite{HouriEtal:2010a}, the others are new.
Let $\alpha\in\Omega^p(M)$,  $\beta\in\Omega^q(M)$ and $p\leq q$. Then the Clifford product expands as
\be
\alpha \beta = \sum_{m=0}^p \frac{ (-1)^{m(p-m) + [m/2]}}{m!} \alpha \cwedge{m}  \beta  \, , \label{usefulProduct1}
\ee
and
\be
\beta \alpha = (-1)^{pq} \sum_{m=0}^p \frac{ (-1)^{m(p-m +1) + [m/2]}}{m!} \alpha \cwedge{m}  \beta  \, .
\label{usefulProduct2}
\ee
As an application, for general inhomogeneous odd forms $\alpha$ and $\beta$ we get\footnote{%
In the sums for which the upper limit is omitted we assume summation over all integer values starting with the lower limit. Number of terms in the sum is finite since the summed expression vanishes for a sufficiently large summation index.}
\be
\bigl[\alpha,\beta\bigr] = 2 \sum_{k=0} \frac{(-1)^k}{(2k)!}\,\alpha\cwedge{2k}\beta\;,\qquad
\bigl[\alpha,\beta\bigr]_+ = 2 \sum_{k=0} \frac{(-1)^k}{(2k{+}1)!}\,\alpha\cwedge{2k+1}\beta\;,\qquad
\label{useful1_odd}
\ee
and for even form $\alpha$ and arbitrary $\beta$
\be
\bigl[\alpha,\beta\bigr] = 2 \sum_{k=0} \frac{(-1)^{k+1}}{(2k{+}1)!}\,\alpha\cwedge{2k+1}\beta\;,\qquad
\bigl[\alpha,\beta\bigr]_+ = 2 \sum_{k=0} \frac{(-1)^k}{(2k)!}\,\alpha\cwedge{2k}\beta\;.
\label{useful1_even}
\ee
These can be translated into graded (anti)-commutator relations which hold for arbitrary forms ${\alpha}$, ${\beta}$
\begin{equation}\label{useful_graded}
  \bgc{\alpha}{\beta} = 2 \sum_{k=0} \frac{(-1)^{k+1}}{(2k{+}1)!}\,(\eta\alpha)\cwedge{2k+1}\beta\;,\qquad
  \bgc{\alpha}{\beta}_+ = 2 \sum_{k=0} \frac{(-1)^k}{(2k)!}\,\alpha\cwedge{2k}\beta\;.
\end{equation}

For $\omega\in\Omega^1(M)$ and $\alpha$, $\beta$ generic inhomogeneous forms it holds
\begin{equation}\label{useful2+3}
\begin{aligned}
( \omega \wedge \alpha ) \cwedge{m} \beta &=
  (-1)^m \omega \wedge \bigl( \alpha \cwedge{m} \beta \bigr)
  + m\, \alpha \cwedge{m{-}1} \bigl( \omega^\sharp \hook \beta \bigr) \, , \\
\alpha \cwedge{m} \bigl( \omega \wedge \beta \bigr)  &=
  \bigl((-1)^\pi \omega\bigr) \wedge \bigl( \alpha \cwedge{m} \beta \bigr)
  + m\, \bigl( \omega^\sharp \hook \alpha \bigr)  \cwedge{m{-}1} \beta  \, .
%\left( \omega \wedge \alpha \right) \cwedge{m} \beta &=& (-1)^m \omega \wedge \left( \alpha \cwedge{m} \beta \right) + m \alpha \cwedge{m-1} \left( \omega^\sharp \hook \beta \right) \label{useful2} \, , \\
%\alpha \cwedge{m} \left( \omega \wedge \beta \right)  &=& (-1)^p \omega \wedge \left( \alpha \cwedge{m} \beta \right) + m \left( \omega^\sharp \hook \alpha \right)  \cwedge{m-1} \beta \label{useful3} \, .
\end{aligned}
\end{equation}
The hook operator acts on the contracted wedge product in the following way:
\begin{equation}\label{useful4}
X \hook \bigl( \alpha \cwedge{m} \beta \bigr) =
  (-1)^m \bigl( X \hook \alpha \bigr) \cwedge{m} \beta
  + \bigl((-1)^\pi\alpha\bigr) \cwedge{m} \bigl( X \hook \beta \bigr)  \, .
\end{equation}

It is useful to generalize the identity $e^a \wedge \left( X_a \hook \alpha \right) = \pi \alpha$ to
\begin{equation}\label{mywedge_hookWedge}
\begin{aligned}
e^a \wedge \Bigl( \bigl( X_a \hook \alpha \bigr) \cwedge{m} \beta \Bigr) &=
  (-1)^m\, \bigl((\pi-m)\alpha\bigr) \cwedge{m} \beta \,, \\
e^a \wedge \Bigl( \alpha \cwedge{m} \bigl( X_a \hook \beta \bigr) \Bigr) &=
  \bigl((-1)^\pi \alpha\bigr) \cwedge{m} \bigl((\pi-m)\beta\bigr) \, .
\end{aligned}
\end{equation}
Using this we can prove
\be
\left( e^b \wedge \alpha \right) \cwedge{m} \left( e_b \wedge \beta \right) = m \left( \dm - p - q + m -1 \right) \alpha \cwedge{m-1} \beta \, . \label{mywedgeContraction}
\ee
The following identity  will be used in App.~\ref{apx:PCKYsym}. Let $\rho$ and $\sigma$ be two arbitrary even-forms, $\omega^a$ arbitrary 1-forms, and $o^a=(\omega^a)^\sharp$ the corresponding vectors. Then
\ba
[\omega^{(a}\wedge \rho] \cwedge{2k} [\omega^{b)}\wedge \sigma]  &=& 2k ( o^a \hook \omega^b) \bigl(\rho \cwedge{2k-1} \sigma\bigr) - 2k \omega^{(a} \wedge \Bigl[o^{b)} \hook \bigl( \rho \cwedge{2k-1} \sigma \bigr)\Bigr]  \nn \\
&& - 2k (2k-1)\bigl( o^{(a} \hook \rho \bigr) \cwedge{2k-2} \bigl( o^{b)} \hook \sigma \bigr) \, . \label{eq:even_contraction_ab}
\ea
This follows from repeatedly applying identities \eqref{useful2+3} and \eqref{useful4}.

The following two useful identities for derivatives of a contracted wedge product were proved in
\cite{HouriEtal:2010a}:
\begin{eqnarray}
\delta(\alpha \cwedge{k}\beta) &=& (-1)^k \delta \alpha \cwedge{k}
\beta - (-1)^p \nabla_a \alpha \cwedge{k} \,(X^a\hook \beta) \nonumber\\
&&+(-1)^p  \alpha \cwedge{k}
\delta \beta - (-1)^k (X^a\hook \alpha)\cwedge{k} \nabla_a \beta\,,\label{A23}\\
d(\alpha \cwedge{k}\beta) &=& (-1)^k d \alpha \cwedge{k}
\beta - (-1)^k k \nabla_a \alpha \cwedge{k-1} (X^a\hook \beta)\nonumber\\
&&+(-1)^p\alpha \cwedge{k}
d\beta - (-1)^p k (X^a\hook \alpha)\cwedge{k-1} \nabla_a \beta \,,\quad\label{A24}
\end{eqnarray}
In particular, when $\alpha$ and $\beta$ are CCKY forms we have
\ba
\delta \left( \alpha \cwedge{m} \beta \right) &=& \left(\dm - p - q + m +1 \right) \left[ \frac{(-1)^m}{\dm - p+1} \delta \alpha \cwedge{m} \beta + \frac{(-1)^p}{\dm - q+1} \alpha \cwedge{m} \delta \beta \right] \, , \label{mywedge_delta}\\
d \left( \alpha \cwedge{m} \beta \right) &=& m \left[ (-1)^p \frac{q-m+1}{\dm - p+1} \delta \alpha \cwedge{m-1} \beta + (-1)^m \frac{p-m+1}{\dm - q+1} \alpha \cwedge{m-1} \delta \beta \right]   \label{mywedge_d} \, .
\ea

\subsection{Hodge duality}\label{apx:HodgeD}

The Hodge dual of a homogeneous ${p}$-form ${\omega}$,
\begin{equation}\label{HDdef}
    *\omega = \frac1{p!}\,\omega\cwedge{p}\eps\;,
\end{equation}
is defined in terms of Levi-Civita tensor ${\eps}$, which is an antisymmetric ${\dm}$-form satisfying
${\displaystyle\eps\cwedge{n}\eps = s\,\dm!}$.
%\begin{equation}\label{HD}
%    \eps\cwedge{n}\eps = s\,\dm!\;.
%\end{equation}
Here, ${s}$ is a product of signs in the metric signature. It follows that
\begin{equation}\label{epsids}
    **\omega = s (-1)^{p(\dm-p)}\omega\;,\quad
    *1 = \eps\;,\quad *\eps = s\;,\quad
    \eps^2 =  (-1)^{[\frac\dm2]}\, s\;.
\end{equation}
Since the contracted wedge product ${\displaystyle\omega\cwedge{k}\eps}$ of the Levi-Civita tensor with a homogeneous ${p}$-form ${\omega}$ is nonzero only for ${k=p}$, using \eqref{usefulProduct1}, \eqref{usefulProduct2} one can write
\begin{equation}\label{HodgeDualClapx}
    *\omega = (-1)^{[\frac p2]}\omega\eps = (-1)^{(n-1)p+[\frac p2]}\eps\omega\;,
\end{equation}
which can be easily generalized for inhomogeneous forms by linearity. As a corollary,
\begin{equation}\label{epscom}
    \eps\omega = (-1)^{(n-1)p}\,\omega\eps\;.
\end{equation}

Applying \eqref{usefulProduct1} on the both sides of identity ${(\eps\alpha)\beta=\eps(\alpha\beta)}$ and comparing the terms with the same degree, one gets for homogeneous forms ${\alpha}$, ${\beta}$ of degrees ${p}$, ${q}$ useful relations
\begin{equation}\label{HDcwedge}
    *\bigl(\alpha\cwedge{k}\beta)
      = \frac{k!}{(q-k)!}(-1)^{k(n-q)}(*\alpha)\cwedge{q{-}k}\beta
      = \frac{k!}{(p-k)!}(-1)^{p(q-k)}\alpha\cwedge{p{-}k}(*\beta)\;.
\end{equation}
Applying both these relation together we obtain
\begin{equation}\label{cwedgeHDs}
    (*\alpha)\cwedge{k}(*\beta) =
    s(-1)^{k(p{+}q)}\frac{k!}{(p{+}q{-}\dm{+}k)!}\,\alpha\cwedge{p{+}q{-}\dm{+}k}\beta\;.
\end{equation}
As a particular cases we can write down the following identities, which hold also for inhomogeneous forms:
\begin{equation}\label{HDwedgehook}
    \eps(e^a\wedge\omega) = (-1)^{n-1} X^a\hook (\eps\omega)\;,\quad
    \eps(X^a\hook\omega) = (-1)^{n-1} e^a\wedge (\eps\omega)\;.
\end{equation}
Taking the covariant derivative of these identities one obtains the known duality of exterior derivative and co-derivative:
\begin{equation}\label{HDddelta}
    \eps(d\omega) = (-1)^{n} \delta (\eps\omega)\;,\quad
    \eps(\delta\omega) = (-1)^{n} d (\eps\omega)\;,
\end{equation}
or in terms of the Hodge dual ${*d\omega = - \delta (*\eta\omega)}$ and ${*\delta\omega = d (*\eta\omega)}$.

\section{Results for the CCKY 2--form tower}\label{apx:PCKYsym}

First, let us mention that Eq. \eqref{mywedge_delta} implies a useful relation
\be
\frac1{{\dm {-} 2j {+}1}}\,\delta h^{(j)}
  =  \frac{j}{\dm-1}\, \delta h \wedge h^{(j-1)}
  = -j\, \xi\wedge h^{(j-1)} \; , \label{eq:delta_hi}
\ee
with ${\xi}$ given by Eq.~\eqref{PCKYsec}.

Now we prove the important properties of the operators ${h^{(j)}}$:
\begin{lemma}\label{lemma_tower1}
Let $h$ be a CCKY 2--form, $h^{(i)}$, $i=1,\dots N-1$ the corresponding tower of CCKY forms \eqref{hj}. Then
\begin{equation}\label{hihj=0}
\bigl[ h^{(i)} , h^{(j)} \bigr] = 0 \;,
\end{equation}
and for arbitrary 1-forms $\omega^a$,
\begin{equation}\label{proofco}
\bigl[\omega^{(a}\wedge h^{(i)}\bigr] \cwedge{2k}\, \bigl[\omega^{b)}\wedge h^{(j)}\bigr]  = 0\;.
\end{equation}
\end{lemma}
\begin{proof}
To prove the first relation we expand it by using \eqref{useful1_even}. Hence we must show that all odd contractions of
$h^{(i)}$ and $h^{(j)}$ vanish,
\be\label{A37}
h^{(i)} \cwedge{2k+1} h^{(j)}=0\,,\quad k=0,\,1,\,2,\,\dots\;.
\ee
Such contractions can be expanded into a sum of a number of terms, each of these terms being a product of `cyclic' sub-terms (i.e. contracted products of ${h}$'s of the kind ${h_{a_1}{}^{a_2}h_{a_2}{}^{a_3}\dots h_{a_{2l}}{}^{a_1}}$) and `linear' sub-terms (of the kind ${h_{a}{}^{a_2}h_{a_2}{}^{a_3}\dots h_{a_{l}b}}$). Since the total number of contractions is ${2k+1}$ in any term there must be at least one linear sub-term with an odd number of contractions. Such a term ${h_{a}{}^{a_2}h_{a_2}{}^{a_3}\dots h_{a_{2l}b}}$ is symmetric in the free indices. However, the contracted wedge \eqref{A37} has all free indices antisymmetrized, hence, the linear term with even number of ${h}$'s must vanish, ${h_{[a|}{}^{a_2}h_{a_2}{}^{a_3}\dots h_{a_{2l}|b]}=0}$.
All terms in the sum thus vanish and the first relation \eqref{hihj=0} is established.

The second statement follows by induction on $k$. The identity is trivial for $k=0$, with any $i$, $j$, and $\omega^a$. Suppose it is true for ${k-1}$. Then by applying \eqref{eq:even_contraction_ab} and using \eqref{A37} we have $[o^a\equiv (\omega^a)^\sharp$]
\begin{equation}
\begin{split}
&[\omega^{(a}\wedge h^{(i)}]  \cwedge{2k} [\omega^{b)}\wedge h^{(j)}]
  = - 2k (2k-1) \Bigl( o^{(a} \hook h^{(i)} \Bigr) \cwedge{2k-2}  \Bigl( o^{b)} \hook  h^{(j)}  \Bigr) \\
  &\qquad\qquad\qquad = -2k\,ij\, (2k-1)\, \Bigl[ \bigl(o^{(a} \hook h\bigr) \wedge h^{(i-1)}  \Bigr]
  \cwedge{2k-2}
  \Bigl[ \bigl(o^{b)} \hook h\bigr) \wedge h^{(j-1)}  \Bigr]\,. \qquad\ \
\end{split}
\end{equation}
Since we can write $o^a \hook h = \tilde{\omega}^a$ the last expression is zero by the induction hypothesis.
\end{proof}

\begin{lemma}\label{lemma_tower2}
Let $h$ be a CCKY 2--form, $h^{(i)}$, $i=1,\dots N-1$ the corresponding tower of CCKY forms \eqref{hj} and $f^{(i)}=*h^{(i)}$ the associated tower of KY forms \eqref{fj}. Then
\begin{equation}\label{algcondhhff}
[e^{(a}\wedge h^{(i)}, e^{b)}\wedge h^{(j)}] = 0  \,,\quad
[X_{(a}\hook f^{(i)}, X_{b)}\hook f^{(j)}]=0      \,.
\end{equation}
Moreover, in an odd dimension, it also holds
\begin{equation}
[X_{(a}\hook f^{(i)}, e_{b)}\wedge h^{(j)}]=0     \,.\label{algcondfh}
\end{equation}
\end{lemma}
\begin{proof}
The first relation in \eqref{algcondhhff} follows from Lemma~\ref{lemma_tower1}. The other can be obtained applying Hodge duality to the first, namely using Eqs. \eqref{HodgeDualClapx}, \eqref{HDwedgehook}, and \eqref{epscom}. The last relation \eqref{algcondfh} follows analogously. However, since the Hodge dual is applied only in the first argument of the commutator, one receives an additional relative sign ${(-1)^{n{-}1}}$ which changes the commutator into anticommutator in even dimensions. Therefore, we need the additional restriction on the dimension.
\end{proof}

Lemma \eqref{lemma_tower2} provides the conditions which guarantee the existence of the Killing--Yano brackets among ${h^{(j)}}$ and ${f^{(j)}}$, respectively. The main result of Sec.~\ref{apx:setcomop} is based on the the fact that these Killing--Yano brackets vanish:
\begin{lemma}\label{lemma_tower3}
Let $h$ be a CCKY 2--form and $h^{(i)}$, $i=1,\dots N-1$ the corresponding tower of CCKY forms \eqref{hj}.
Then
\begin{equation}\label{KYbrhhApp}
[h^{(i)}, h^{(j)}]_\KY = 0  \;.
\end{equation}
\end{lemma}
\begin{proof}
We use Eq.~\eqref{KYbracketsHomMM} from Prop.~\ref{PropKYbrhom} to evaluate the bracket. In odd dimension we obtain (up to a numerical prefactor)
\begin{equation}
  [h^{(i)}, h^{(j)}]_\KY \propto  \delta\bigl( h^{(i)}\cwedge{2i+2j-\dm} h^{(j)} \bigr) = 0
\end{equation}
which vanishes thanks to \eqref{hihj=0}. In even dimensions one gets
\be
 [h^{(i)}, h^{(j)}]_\KY \propto
    -\frac{1}{\dm - 2i +1} \delta h^{(i)} \cwedge{m_*} h^{(j)}
    + \frac{1}{\dm-2j+1} h^{(i)} \cwedge{m_*} \delta h^{(j)} \, ,
\ee
where $m_* = 2i + 2j - \dm -1$ is such that we cannot use Eq.~\eqref{mywedge_delta} to reduce the non-trivial term to a total co-differential. We start using Eq.~\eqref{eq:delta_hi} to rewrite the term as
\be
  [h^{(i)}, h^{(j)}]_\KY \propto
   i\, \bigl( \xi \wedge h^{(i{-}1)} \bigr) \cwedge{m_*} h^{(j)}
   - j\, h^{(i)} \cwedge{m_*} \bigl( \xi \wedge h^{(j{-}1)} \bigr) \, .
\ee
We can expand the first term using identities \eqref{useful2+3}, \eqref{A37}, \eqref{useful4}, and again \eqref{useful2+3} as follows
\begin{equation}
\begin{split}
&i \bigl( \xi \wedge h^{(i{-}1)} \bigr) \cwedge{m_*} h^{(j)}
  = -i\, \xi \wedge \bigl(  h^{(i{-}1)} \cwedge{m_*}  h^{(j)} \bigr)
    + i m_{*}\,  h^{(i{-}1)} \cwedge{m_* {-}1} \bigl( \xi^\sharp \hook  h^{(j)} \bigr) \\
&\qquad = m_* i j\,  h^{(i{-}1)} \cwedge{m_*{-}1} \bigl( (\xi^\sharp \hook h) \wedge  h^{(j{-}1)} \bigr) \\
&\qquad = m_* ij \bigl(\xi^\sharp \hook h \bigr) \wedge \bigl(  h^{(i{-}1)} \cwedge{m_*{-}1}  h^{(j{-}1)} \bigr)
   + m_* (m_* {-}1) ij \bigl( (\xi^\sharp \hook h) \hook \wedge h^{(i{-}1)} \bigr) \cwedge{m_*{-}2}  h^{(j{-}1)} \, .
\end{split}\raisetag{9ex}
\end{equation}
A similar expansion of the other term gives
\begin{equation}
\begin{split}
j\, h^{(i)}  \cwedge{m_*} \bigl( \xi \wedge h^{(j{-}1)} \bigr)
  &= m_* ij\, (\xi^\sharp \hook h ) \wedge \bigl(  h^{(i{-}1)} \cwedge{m_*{-}1}  h^{(j{-}1)} \bigr) \\
&\quad
    + m_* (m_*{-}1) ij  h^{(i{-}1)} \cwedge{m_*{-}2}  \bigl( (\xi^\sharp \hook h)\hook h^{(j{-}1)} \bigr) \, .
\end{split}
\end{equation}
Summing these two together we get
\begin{equation}
\begin{split}
[h^{(i)}, h^{(j)}]_\KY &\propto
% m_* (m_* {-}1) ij
  \bigl( (\xi^\sharp \hook h) \hook \wedge h^{(i{-}1)} \bigr) \cwedge{m_*{-}2}  h^{(j{-}1)}
  - h^{(i{-}1)} \cwedge{m_*{-}2}  \bigl( (\xi^\sharp \hook h)\hook h^{(j{-}1)} \bigr) \\
& = -(\xi^\sharp \hook h)\hook
  \bigl( h^{(i{-}1)} \cwedge{m_*{-}2}  h^{(j{-}1)} \bigr) = 0\;.
\end{split}
\end{equation}
Here we used Eq.~\eqref{useful4} and finally again Eq.~\eqref{A37}, remembering that ${m_*-2}$ is odd.
\end{proof}

%%%%%%%%%%%%%%%%%%%%%%%%%%%%%%%%%%%%%%%%%%%%%%%%%%%%%%%%%%

%%%%%%%%%%%%%%%%%%%%%%%%%%%%%%%%%%%%%%%%%%%%%%%%%%%%%%%%%%%%%%%%%%%%%%%%%%
%%%%%%%%%%%%%%%%%%%%%%%%%%%%%%%%%%%%%%%%%%%%%%%%%%%%%%%%%%%%%%%%%%%%%%%%%%

%%%%%%%%%%%%%%%%%%%%%%%%%%%%%%%%%%%%%%%%%%%%%%%%%%%%%%%%%%%%%%%%%%%%%%%%%%
%%%%%%%%%%%%%%%%%%%%%%%%%%%%%%%%%%%%%%%%%%%%%%%%%%%%%%%%%%%%%%%%%%%%%%%%%%

%\bibliography{Databaze}
%\bibliographystyle{JHEP}

\providecommand{\href}[2]{#2}\begingroup\raggedright
\end{document}